\documentclass[11pt,a4paper]{article}

\usepackage{graphicx}
\graphicspath{{figures/}}
\usepackage[english]{babel}
\usepackage{latexsym}
\usepackage{url}
\usepackage{amsmath}
\usepackage{amsthm,amsfonts}
\usepackage{dsfont}
\usepackage{ifthen}
\usepackage{fancybox}
\usepackage{microtype,stmaryrd}

\usepackage[charter]{mathdesign}

\usepackage{color}
\definecolor{blu3}{rgb}{.1,.0,.4}
\usepackage{hyperref}
\hypersetup{colorlinks=true, linkcolor=blu3, urlcolor=blu3, citecolor=blu3, pdfpagemode=UseNone, pdfstartview=} 

\usepackage[margin=1.1in]{geometry}

\newtheorem{theorem}{Theorem}
\newtheorem{corollary}[theorem]{Corollary}
\newtheorem{lemma}[theorem]{Lemma}

\newcommand{\RR}{\ensuremath{\mathbb R}}  
\newcommand{\PP}{\ensuremath{\mathcal{P}}}  

\DeclareMathOperator{\loglog}{loglog}

\DeclareMathOperator{\exterior}{ext}
\DeclareMathOperator{\partsum}{partsum}
\DeclareMathOperator{\diam}{diam}
\DeclareMathOperator{\IGL}{IGL}
\DeclareMathOperator{\ISW}{ISW}
\DeclareMathOperator{\cell}{cell}

\newcommand\eps{\varepsilon}

\def\DEF#1{\textbf{\emph{#1}}}

\def\dart#1#2{#1\mathord\shortrightarrow#2}

\begin{document}

\title{Computing the Inverse Geodesic Length\\
		in Planar Graphs and Graphs of Bounded Treewidth}

\author{Sergio Cabello\thanks{Faculty of Mathematics and Physics, University of Ljubljana, Slovenia;
				Institute of Mathematics, Physics and Mechanics, Slovenia.
                Supported by the Slovenian Research Agency (P1-0297, J1-9109, J1-8130, J1-8155, J1-1693, J1-2452).
				Email address: sergio.cabello@fmf.uni-lj.si}}

\maketitle

\begin{abstract}
	The inverse geodesic length of a graph $G$ is the sum of the inverse of the distances
	between all pairs of distinct vertices of $G$. In some domains
	it is known as the Harary index or the global efficiency of the graph.
	We show that, if $G$ is planar and has $n$ vertices, then
	the inverse geodesic length of $G$ can be computed in roughly $O(n^{9/5})$ time.
	We also show that, if $G$ has $n$ vertices and treewidth at most $k$, 
	then the inverse geodesic length of $G$ can be computed in $O(n \log^{O(k)}n)$ time.
	In both cases we use techniques developed for computing the sum of the distances,
	which does not have ``inverse'' component, 
	together with batched evaluations of rational functions.

\end{abstract}

\section{Introduction}
Let $G=(V,E)$ be an undirected graph with $n$ vertices and 
(abstract) positive edge-lengths $\lambda\colon E\rightarrow \RR_{>0}$.
The length of a walk in $G$ is the sum of the edge-lengths along
the walk. 
The \DEF{distance} between two vertices $u$ and $v$ of $G$, denoted
by $d_G(u,v)$, is the minimum length over all paths in $G$ from $u$ to $v$. 
A particularly important case is when $\lambda(e)=1$ for all edges $e\in E$.
In this case, the distance between two vertices $u$ and $v$ 
is the minimum number of edges over the $u$-$v$ paths in $G$.
(Usually this is called the \emph{unweighted} case or 
\emph{unit-length} case.)

In this paper we are interested in the \DEF{inverse geodesic length} of $G$, 
defined as
\[
	\IGL(G) ~=~ \sum_{ uv \in \binom{V}{2}} \frac{1}{d_G(u,v)}\, ,
\]
where $\binom{V}{2}$ denotes all the unordered pairs of vertices of $G$.
In this definition we use the convention that $1/d_G(u,v) =0$ when there is no
path from $u$ to $v$.
As a consequence, $\IGL(G)$ is equal to the sum of the inverse geodesic length
over its connected components.

The inverse geodesic length has been considered in different contexts.
In Chemical Graph Theory it is considered for unit-length edges
and called the \DEF{Harary index} of graph $G$,
a common topological index~\cite{TopIndices99,PNTN1993};
sometimes it is defined as twice $\IGL(G)$.
In the context of Network Analysis, the inverse geodesic length 
goes under the name of \DEF{efficiency}, a concept introduced by
Latora and Marchiori~\cite{Latora01,Latora2003} that has been heavily used. 
In this context, $\IGL(G)$ is usually normalized dividing by $\binom n2$, 
that is, the \emph{average} inverse geodesic length is considered.
The author learned the concept with the algorithmic 
work of Gaspers and Lau~\cite{GaspersL19}
and refers to the more careful discussion about related work therein.

The problem is related to but different from the computation of the \DEF{Wiener index} $W(G)$
and the \DEF{diameter} $\diam(G)$, defined as
\[
	W(G)~=~\sum_{ uv \in \binom{V}{2}} d_G(u,v), ~~~~~~~~
	\diam(G)~=~\max\{ d_G(u,v)\mid uv \in \binom{V}{2}\} .
\]

\paragraph{Our contribution.} 
We present algorithms to efficiently
compute $\IGL(G)$ for two standard graph classes. The challenge to obtain
efficient algorithms is to avoid computing the distance between
all pairs of vertices. More precisely, 
we show the following results, where $n$ denotes the number of vertices of the graph:
\begin{itemize}
	\item For each fixed $k\geq 1$, the inverse geodesic length of graphs 
		with treewidth at most $k$ can be computed in 
		$O(n \log^{k+2}n \loglog n)$ time. Thus, the running
		time is near-linear for graphs of bounded treewidth.
	\item When we consider the running time parameterized by $n$ and the treewidth $k$,
		the inverse geodesic length can be computed in $n^{1+\eps} 2^{O_\eps(k)}$ time,
		for any $\eps>0$.\footnote{The 
		constants hidden in the $O_\eps$-notation depend on $\eps$.} 
	\item For planar graphs, the inverse geodesic length can be 
		computed in $\tilde O(n^{9/5})$ time\footnote{Here and in the rest of the paper we
		use the notation $\tilde O$ to hide logarithmic factors.}.
\end{itemize}

To achieve these results we build on the techniques used to compute 
the Wiener index and the diameter in planar graphs~\cite{Cabello2019,GawrychowskiKMS21}, namely the use
of additively-weighted Voronoi diagrams in planar graphs, 
and in graphs with small treewidth~\cite{AbboudWW16,BringmannHM20,CabelloK09},
namely the use of orthogonal range searching.

Previous papers computing the Wiener index or the diameter rely on the fact that
the length of a shortest path is the sum of the lengths of subpaths that compose it. 
This is not true for the inverse of distances.
More precisely, those works decompose the problem into subproblems
described by a triple $(a,s,U)$,
where $a$ and $s$ are vertices, $U$ is a subset of vertices, and the shortest
path from $a$ to each vertex of $U$ goes through vertex $s$. 
Moreover, a preprocessing step is used to construct a data structure
so that $|U|$, $\sum_{u\in U} d_G(s,u)$ and $\max_{u\in U} d_G(s,u)$ 
are obtained in sublinear time (in $|U|$).
This preprocessing pays off because there are several subproblems
with triples of the form $(\cdot,s,U)$ for the same $s$ and $U$.
Then we use that
\begin{align*}
	\sum_{u\in U} d_G(a,u) ~&=~ \sum_{u\in U} \Bigl( d_G(a,s)+d_G(s,u)\Bigr) ~=~
	|U|\cdot d_G(a,s)+\sum_{u\in U} d_G(s,u) \, ~~\text{ and}\\
	\max_{u\in U} d_G(a,u) ~&=~ d_G(a,s)+\max_{u\in U} d_G(s,u)\,.
\end{align*}
There is no such simple decomposition of 
\[
	\sum_{u\in U} \frac{1}{d_G(a,u)} ~=~ \sum_{u\in U} \frac{1}{d_G(a,s)+d_G(s,u)}
\]
that would be useful for $\IGL(G)$. That is the main obstacle to 
apply the techniques.

Our main new technique is to associate to $(\cdot,s,U)$ a rational function
\[
	\rho_{s,U}(x) ~=~ \sum_{u\in U} \frac{1}{x+d_G(s,u)}.
\]
Then, for the triple $(a,s,U)$, we are interested in evaluating $\rho_{s,U}(d_G(a,s))$.
We use computer algebra to make such evaluations at several values in near-liner time, 
as done by Aronov, Katz and Moroz~\cite{AronovK18,MorozA16} in completely different settings.
The key insight is that $\rho_{s,U}(x)$ can be computed efficiently
because, after expansion, it is a rational function
whose numerator and denominator have degree $|U|$.
Note that all the evaluations of $\rho_{s,U}(x)$ are computed simultaneously, 
which is also a change in the approach compared to previous works for the Wiener index or the diameter. 
Thus, we do not use a data structure to handle each triple $(a,s,U)$, but
we treat all the triples of the form $(\cdot,s,U)$ batched.

We believe that the use of tools from Computational Geometry 
and Computer Algebra
for solving graph-distance problems, such as computing the $\IGL$,
shows an appealing interaction between different areas.

\paragraph{Comparison to related work.}
The efficient computation of $\IGL$ for graphs with small treewidth was considered by 
Gaspers and Lau~\cite{GaspersL19}. 
They show that $\IGL(G)$ for graphs with unit-length edges and treewidth at most $k$ 
can be computed in $2^{O(k)} n^{3/2+\varepsilon}$.
In the case of unit-length trees the running time is reduced to $O(n\log^2 n)$.
Their approach is based on computing the distance distribution:
for any positive integer $P$,
one can compute the number of pairs of vertices at distance $i$ for 
$i=1,\dots,P$ in $2^{O(k)} n^{1+\varepsilon} \sqrt{P}$ time.
In the case of trees the distance distribution can be computed 
$O(n \log n + P \log^2 n)$ time.
From the distance distribution it is easy to compute the $\IGL(G)$ in linear time.

We improve the result of Gaspers and Lau in two ways. First, we reduce the degree
of the main polynomial term of the running time from roughly $3/2$ to $1$ for graphs of small
treewidth. Second, our results
work for arbitrary (positive) weighted graphs, not only unit length (or small integers).
For the case of unit-length trees, their result is stronger than ours.
Our algorithm does not compute the distance distribution. 
In fact, for arbitrary weights it would be too large. 
Their algorithm for computing the distance distribution relies on the fast product of polynomials.
They do not use rational functions, which is a key insight in our approach.
They use the degree to encode distances, so the edge lengths must be small positive integers,
while in our approach the degree of the rational functions depends (implicitly) 
on the number of vertices.

For graphs with small treewidth, our algorithm to compute the $\IGL$ 
has a few more logarithmic factors than the algorithms 
to compute the Wiener index or the diameter~\cite{BringmannHM20,CabelloK09}.
This slight increase in the running time is due to the 
manipulation and evaluation of rational functions.

\medskip

We are not aware of any previous algorithmic result for computing the $\IGL$
of planar graphs. Since the distance between all pairs of vertices in a 
planar graph can be computed in $O(n^2)$ time~\cite{f-pgdap-91,hkrs-fspap-97},
the $\IGL$ of planar graphs can be computed in $O(n^2)$ time.
We provide the first algorithm computing the $\IGL$ of planar graphs in
subquadratic time, even in the special case of unit-length edges.

For planar graphs, the running time of computing the $\IGL$ is 
$\tilde O(n^{9/5})$, while the Wiener index and the diameter
can be computed in $\tilde O(n^{5/3})$ time
with the algorithm of Gawrychowski et al.~\cite{GawrychowskiKMS21}.
The difference is quite technical and difficult to explain at
this level of detail. For readers familiar with that work,
we can point out that, to apply our new approach,
the dynamic tree used in~\cite{GawrychowskiKMS21} 
to encode the bisectors should be combined with the rational functions. 
The natural way to do this would be to associate a rational function to each node 
of the dynamic tree,
and that leads to an overhead of $O(r)$ when working with a piece
of an $r$-division. A more careful treatment of the 
dynamic tree may lead to an improvement, but we do not see how. 
The approach we use to manipulate (the inverse of) the distances 
for each piece is similar to the one used in Cabello~\cite{Cabello2019},
where the bisectors are computed and manipulated explicitly.
It should be noted that, at this point, the bottleneck in our algorithm
is not the computation of the Voronoi diagram,
but the time to manipulate the outcome of the Voronoi diagram.

\medskip

For arbitrary graphs, there is no constant $\delta_0>0$ such
that the $\IGL$ can be computed in $O(n^{2-\delta_0})$ time, unless
the strong exponential time hypothesis (SETH) fails. This holds
also for sparse graphs.
Indeed, Roditty and Vassilevska Williams~\cite{RodittyW13}
show that, for arbitrary graphs with $n$ vertices and $O(n)$ edges, 
one cannot compute the diameter in $O(n^{2-\delta_0})$ time, 
for some constant $\delta_0>0$, unless the SETH fails.
In fact, their proof shows that for undirected, unweighted graphs
we cannot distinguish in $O(n^{2-\delta_0})$ time
between sparse graphs that have diameter $2$ or larger.
Since a unit-length graph $G=(V,E)$ has diameter $2$ if and only if  
\[
	\IGL(G) ~=~ \frac 12 \, \sum_{u\in V} \Bigl( \deg_G(u)+(n-1-\deg_G(u))\tfrac{1}{2} \Bigr)
	~=~ 
	\frac{n(n-1)}{2}-\frac{3|E(G)|}{4},
\]	
we cannot compute the $\IGL$ in $O(n^{2-\delta_0})$ time, unless SETH fails.

\paragraph{Computation model.} 
We assume a model of computation where each arithmetic operation 
in the ring generated by the edge-lengths takes constant time 
and can be carried out exactly. 
Since we are using the Fast Fourier Transform (FFT), the actual running time depends
on whether we assume that the ring has primitive roots of the unit available. 
The assumptions are explained in more detail in Section~\ref{sec:algebra}.

The output and several intermediate computations are represented as fractions, 
that is, as pairs $(a,b)$ of numbers representing the fraction $a/b$. 
The numerators and the denominators of these fractions are in the ring generated 
by the edge-lengths. Thus, if we assume that the input edge-lengths are integers, 
the fractions will have integer numerators and denominators, 
if the input edge-lengths are rationals, 
the fractions will have rational numerators and denominators,
and if the edge-lengths are arbitrary real numbers, then the fractions
will have numerators and denominators with real numbers.

The model of computation used by Gaspers and Lau~\cite{GaspersL19} is Word RAM,
which is a weaker model and thus a stronger result. 
They can use this model because they assume graphs with small integer
edge-weights and they compute, for each possible distance, 
the number of pairs of nodes at such distance.
After this, one still has to compute $\IGL(G)$ explicitly
and it is not obvious how to bound the bit-length of the computation.
In our approach, we should also bound the bit-length of
intermediate integers appearing through the computation. 
This is not straightforward. To get a feeling of this endeavor, 
note for example that we should understand the number of bits needed 
to write expressions like the Harmonic numbers 
$H_n=\sum_{i=1}^n \tfrac 1i$ exactly as a single fraction $a_n/b_n$.

\paragraph{Roadmap.}
We provide some common preliminaries in the next section.
In Section~\ref{sec:treewidth} we provide the algorithm for graphs with small or bounded
treewidth, while in Section~\ref{sec:planar} we provide the algorithm for planar graphs.
Sections~\ref{sec:treewidth} and~\ref{sec:planar} are independent and can be 
read in any order. The structure of the exposition in both cases is parallel.
We suggest to the readers to start with the section dedicated to graphs that
are closer to their expertise.

\section{Preliminaries}
\label{sec:preliminaries}

For each positive integer $k$ we use the notation $[k]= \{ 1,\dots,k\}$.

Let $G$ be a graph with positive edge lengths. Since the graph $G$ under
consideration will always be clear from the context, 
we will often drop the dependency on $G$ from the notation.
(We do keep using $d_G(\cdot,\cdot)$ for the distance.)

Let $A$ and $B$ be disjoint subsets of vertices in $G$. We define
\[
	\IGL(A,B) ~=~ \sum_{a\in A, b\in B} \frac{1}{d_G(a,b)}.
\]
For the special case when one of the sets has a single element, for example $B=\{b\}$,
we write
\[
	\IGL(A,b) ~=~ \IGL (A,\{ b \}) ~=~ \sum_{a\in A}\frac{1}{d_G(a,b)}.
\]

A \DEF{separation} in $G$ is a partition of $V(G)$ into pairwise disjoint sets of vertices
$A,B,S$ such that there is no edge of $G$ with one vertex in $A$ and the other vertex in $B$.
The set $S$ is the \DEF{separator}. We will systematically use $a$ for vertices in $A$,
$b$ for vertices in $B$, and $s$ for vertices in $S$. Note that any path from $a\in A$ to $b\in B$
must necessarily pass through a vertex of $S$. In particular,
\[
	\forall a\in A,\, b\in B:~~~ d(a,b)=\min \{ d_G(a,s)+d_G(s,b) \mid s\in S\}.
\]

We use $\sqcup$ to denote the \DEF{disjoint union}. Thus, if we write
$A=\sqcup_{i\in I} A_i$, this means that the sets in $\{A_i \mid i\in I\}$ 
are pairwise disjoint and $A=\cup_{i\in I} A_i$.

\subsection{Algebraic computations.}
\label{sec:algebra}
A polynomial $P(x)=\sum_{i=0}^d a_i x^i$ given in \DEF{coefficient representation}
is the list of coefficients $(a_0,\dots, a_d)$ of length $1+\deg(P)$.
A rational function $R(x)$ is a function $R(x)=\frac{P(x)}{Q(x)}$ where $P(x)$ and $Q(x)$
are polynomials. The coefficient representation of such rational function
is just the coefficient representation of $P(x)$ and $R(x)$. 
The following result is obtained by
using fast multiplication of polynomials followed by fast 
multipoint evaluation of polynomials. We refer to the book by 
von zur Gathen and Gerhard~\cite[Chapters~8 and~10]{book_0031325} 
for a comprehensive treatment of the tools of Computer Algebra that we will use.

\begin{lemma}[Corollary A.4 in Aronov and Katz~\cite{AronovK18}; 
	see also Lemmas 6 and 7 in Moroz and Aronov~\cite{MorozA16}]
\label{lem:multipoint_evaluation}
	Given a set of $n$ rational functions $R_i(x)=P_i(x)/Q_i(x)$ of constant degree each,
	in coefficient representation, and a set of $m$ values $x_j$, 
	one can compute the $m$ values $\sum_i R_i(x_j)$ in time $O((n + m) \log^2 n \loglog n)$.
\end{lemma}
We provide the proof to be able to discuss the model of computation in more detail.
\begin{proof}
	As a first step, we compute polynomials $P(x)$ and $Q(x)$ such that 
	\[
		\frac{P(x)}{Q(x)} ~=~ \sum_{i=1}^n R_i(x) ~=~ \sum_{i=1}^n \frac{P_i(x)}{Q_i(x)}.
	\]
	More precisely, we will compute 
	\[
		P(x)~=~ \sum_{i=1}^n \left( P_i(x) \cdot \prod_{j\neq i} Q_j(x)\right) ~~~~\text{and}~~~~
		Q(x)~=~ \prod_{i=1}^n Q_i(x).
	\]
	We do this using a recursive algorithm.
	If $n=1$, we then have $P(x)=P_1(x)$ and $Q(x)=Q_1(x)$.
	If $n\ge 2$, we compute, for each $j=1,\dots, \lfloor n/2\rfloor$,
	the rational function 
	\[
		\tilde R_j(x) ~:=~ R_{2j-1}(x)+ R_{2j}(x) ~=~ 
		\frac{P_{2j-1}(x)Q_{2j}(x)+P_{2j}(x)Q_{2j-1}(x)}{Q_{2j-1}(x) Q_{2j}(x)}
	\]
	in coefficient representation.
	If $n$ is odd we set $\tilde R_{\lceil n/2\rceil}(x)=R_n(x)$.
	We then compute recursively the rational function $\sum_{j=1}^{\lceil n/2\rceil} \tilde R_j(x)$
	and return it.
	
	To analyze the running time note that, after $k$ levels of the recursion, 
	we have $O(n/2^k)$ rational functions, each
	with a denominator and a numerator of degree $O(2^k)$ and coefficients in the ring
	generated by the coefficients of $\cup_i \{P_i(x), Q_i(x) \}$.
	Let $T_{\rm prod}(d)$ be the time needed to multiply two polynomials of degree $d$ in coefficient representation,
	assuming that each arithmetic operation in a ring that contains the coefficients 
	of the polynomials takes constant time.
	Using the Fast Fourier Transform we have $T_{\rm prod}(d)=O(d \log d \loglog d)$; 
	see~\cite[Theorem~8.23]{book_0031325}.
	The time needed to compute $\frac{P(x)}{Q(x)} ~=~ \sum_i R_i(x)$ is then
	\begin{align*}
		O\left( \sum_{k=1}^{\lceil\log_2 n \rceil} \left( \frac{n}{2^k}\cdot \bigl( 2^k + T_{\rm prod}(2^k)\bigr) \right) \right) ~&=~
		O(n\log n) +  O(n)\cdot \sum_{k=1}^{\lceil\log_2 n \rceil} \left( (2^k \log 2^k \loglog 2^k) /2^k \right) \\ 
&=~
		O(n)\cdot \sum_{k=1}^{\lceil\log_2 n \rceil} (k \log k) \\ 
&=~ 
		O(n \log^2 n \loglog n).	
	\end{align*}
	
	Once we have the polynomials $P(x)$ and $Q(x)$ such that 
	$\frac{P(x)}{Q(x)} ~=~ \sum_i R_i(x)$ in coefficient
	representation, we evaluate $P(x)$ and $Q(x)$ at $x=x_j$ for all $j$.
	Note that $P(x)$ and $Q(x)$ have degree $O(n)$.
	Evaluating a polynomial of degree $d$ at $d$ points takes $O\left(T_{\rm prod}(d) \cdot \log d\right)$ time;
	see~\cite[Corollary~10.8]{book_0031325}.
	To evaluate $P(x)$ and $Q(x)$ at $x=x_j$ for all $j$, we make $\lceil m/n\rceil$ groups of $n$
	values each, and evaluate at each group. The running time for this is
	\begin{align*}
		O(1+m/n)O\left(T_{\rm prod}(d) \cdot \log d\right) ~&=~ O(1+m/n)\cdot O(n \log^2 n \loglog n) 
		\\ &=~ O((n+m) \log^2 n \loglog n).
	\end{align*}
	
	\emph{Note}: the $\loglog n$ term may seem superfluous, but this depends on the assumptions.
	If we assume that the ring supports the FFT, that is, we can manipulate primitive 
	roots of the unity in constant time, and divisions can be carried out exactly,
	then $T_{\rm prod}(d)=O(d \log d)$; see~\cite[Theorem~8.18]{book_0031325}.
	In this case the $\loglog n$ factor in the final running time disappears.
	If we are assuming that the coefficients and the values $x_j$ are arbitrary real numbers, 
	this assumption would mean that arithmetic operations involving the complex
	roots of the unit $\cos(\pi/2^k)+ i \sin(\pi/2^k)$ can be manipulated exactly in constant time.
	In our stated running times we are making the more conservative assumption that only operations
	in the ring generated by $\{ d_G(u,v)\mid u,v\in V(G) \}$ are available, and divisions
	are not performed.
\end{proof}

All our use of computer algebra is encoded in Lemma~\ref{lem:multipoint_evaluation}. 
More precisely, we will be using Lemma~\ref{lem:multipoint_evaluation}
for rational functions of the form
\[
	R_i(x)~=~ \frac{1}{d_G(u_i,v_i)+x} ~~~~\text{ for some $u_i,v_i\in V(G)$}
\]
and we will be evaluating $\sum_i R_i(x)$ at values of the form $x=d_G(u_j,v_j)$
for some $u_j,v_j\in V(G)$.
Thus, looking into the proof of Lemma~\ref{lem:multipoint_evaluation} and the 
proofs in~\cite[Chapters~8 and~10]{book_0031325}, we see that we are making
arithmetic operations in the ring generated by $\{ d_G(u,v)\mid u,v\in V(G) \}$,
without computing inverses and without ever simplifying fractions. 
Thus, we carry exact results, assuming that arithmetic operations are performed exactly.
The running times are expressed assuming that each such arithmetic operation in the ring
takes constant time.

\section{Bounded treewidth}
\label{sec:treewidth}

In this section we provide the near-linear time algorithm for graphs with
bounded treewidth.
First we review what we need for the usual orthogonal range searching queries, and
then we extend this to provide our new, key data structure for so-called inverse shifted queries.
We describe how to compute the interaction across a separator,
and finally apply the divide-and-conquer approach of Cabello and Knauer~\cite{CabelloK09}
to derive the final result. 

\subsection{Preliminaries on orthogonal range searching}
In our algorithm for small treewidth we will use the range trees used for
orthogonal range searching. They are explained in the textbook~\cite[Chapter 5]{bkos-08} 
or the survey~\cite{ae-survey}, but the usual analysis in Computational Geometry
assumes constant dimension. We use the analysis by
Bringmann et al.~\cite{BringmannHM20}, which uses the dimension $d$ as a parameter
and is based on the results of Monier~\cite{Monier80}.

A \DEF{box} in $\RR^d$ is the product of intervals,
some of them possibly infinite and some of them possibly open or closed.
(We only consider such axis-parallel boxes in this paper.)
We will use $B(n,d)=\binom{d+\lceil\log n\rceil}{d}$ to bound the time and size 
of the range trees. First we note a bound on $B(n,d)$, and then provide a multivariate bound
on range trees. Note that the stated bounds do \emph{not} assume that $d$ is constant.

\begin{lemma}[Lemma 5 in Bringmann et al.~\cite{BringmannHM20}]
\label{lem:B(n,d)}
$B(n,d)=O(\log^d n)$ and $B(n,d)=n^{\eps}2^{O_\eps(d)}$ for each $\eps>0$.
\end{lemma}

\begin{theorem}
\label{thm:rangesearching}
	Given a set $P$ of $n$ points in $\RR^d$, there is a family of 
	sets $\mathcal{P}=\{ P_i \mid i\in I\}$ and a data structure with
	the following properties: 
	\begin{itemize}
		\item $P_i\subseteq P$ for each $P_i\in \mathcal{P}$;
		\item all the sets of $\mathcal {P}$ together have $O(nd\cdot B(n,d))$ points, counting with multiplicity;
			that is, $\sum_{P_i\in \mathcal{P}} |P_i|=O(nd\cdot B(n,d))$;
		\item for each box $R\subset \RR^d$,
			the data structure finds in $O(2^d B(n,d))$ time indices $I_R\subset I$ such that 
			$|I_R|= O(2^d B(n,d))$
			and $P\cap R = \bigsqcup_{i\in I_R} P_i$;
		\item the family $\mathcal{P}$ and the data structure 
			can be computed in $O(nd\cdot B(n,d))$ time.
	\end{itemize}
\end{theorem}
\begin{proof}
	We consider the following variant of range trees.
	We assume that the reader is familiar with range trees and provide only
	a sketchy description. 
	We describe a tree $T(d,Q)$ that is defined recursively on the dimension $d$ and 
	size of the point set $Q\subset \RR^d$ under consideration.
	When the dimension $d$ is $1$, we make a balanced binary
	search tree $T(1,Q)$ that stores the elements of $Q$ at the leaves.
	
	When $d>1$ and $|Q|>1$, we split $Q$ into two sets $Q_\ell$ and $Q_r$ of roughly the same size,
	such that each element of $Q_\ell$ has smaller $d$-coordinate than each element
	of $Q_r$. We recursively construct $T(d,Q_\ell)$, $T(d,Q_r)$ and $T(d-1,Q)$.
	Finally, we make a node, the root of $T(d,Q)$,
	whose left child is the root $T(d,Q_\ell)$,
	its right child is the root of $T(d,Q_r)$, and its associated data structure
	is $T(d-1,Q)$.
	When $d>1$ and $|Q|=1$, we make a node without children and with a pointer to the associated
	data structure $T(d-1,Q)$, which is built recursively.
	
	The data structure is $T(d,P)$.
	For each node $v$ of $T(d,P)$, let $P(v)$ denote the set
	of points stored in $T(d,P)$ under $v$. This is a so-called \emph{canonical subset}.
	(In fact it suffices to consider nodes $v$ in trees $T(1,Q)$, but that is not
	really relevant in our discussion.)
	For each node $v$ of $T(d,P)$, we add the canonical subset $P(v)$
	to $\mathcal{P}$.
	
	The analysis by Bringmann et al.~\cite[Lemma 6]{BringmannHM20} shows
	that 
	\[
		\sum_{v} |P(v)| ~=~ O(nd\cdot B(n,d)),
	\]
	where the sum is over all the nodes $v$ of $T(d,P)$.
	This gives an upper bound on the size of the data structure and
	the construction time.
	(For the construction, instead of computing medians, as Bringmann et al.~suggest,
	it is better to
	sort the points once in each dimension and pass to the recursive calls
	the points sorted in $d$ lists, one for each dimension.)

	The analysis of Bringmann et al.~for queries~\cite[Lemma 8]{BringmannHM20}
	identifies for each box $R$
	the nodes of the tree such that $P\cap R$ is the disjoint union of
	$O(2^d B(n,d))$ canonical subsets from $\mathcal{P}$ and they can
	be found in $O(2^d B(n,d))$ time. 
\end{proof}

\subsection{Inverse shifted queries}

Let $P$ be a set of $n$ points in $\RR^d$ and 
assume that each point $p$ of $P$ has a positive weight $\omega(p)$.
An \DEF{inverse shifted-weight query} (ISW query) is specified
by a pair $(R,\delta)$, where $R$ is an (axis-parallel) box and $\delta$ is
a non-negative value, the \emph{shift}.
We want to compute 
\[
	\ISW(R,\delta) ~=~ \sum_{p\in P\cap R} \frac{1}{\delta+\omega(p)} \, .
\]
Since $\omega(p)$ is positive and $\delta\geq 0$, no fraction has $0$ in the denominator
and $\ISW(R,\delta)$ is well defined.

We do not know how to answer ISW queries online, but we can solve the batched
version.

\begin{theorem}
\label{thm:ISW}
	Let $P$ be a set of $n$ points in $\RR^d$ and
	assume that each point $p\in P$ has a positive weight $\omega(p)$.
	Consider $\ISW$ queries specified by $(R_1,\delta_1),\dots, (R_m,\delta_m)$,
	where each $R_j$ is a box and $\delta_j\ge 0$. 
	The values $\ISW(R_1,\delta_1),\dots, \ISW(R_m,\delta_m)$ 
	can be computed in $O((n+m)2^d \cdot B(n,d)\cdot \log^2 n \loglog n )$ time.
\end{theorem}
\begin{proof}
	We compute the family $\mathcal{P} =\{ P_i\mid i\in I\}$ and
	the data structure given in Theorem~\ref{thm:rangesearching}.
	This takes $O(nd\cdot B(n,d))= O(n2^d\cdot B(n,d))$ time.
	For each $R_j$ we use the data structure to compute the 
	set $I_j$ of $O(2^d \cdot B(n,d))$ indices such that 
	$P\cap R_j = \bigsqcup_{i\in I_j} P_i$. For all $j\in [m]$ together
	this takes $O(m2^d \cdot B(n,d))$ time.
 
	Define the pairs $\Pi =\{ (i,j)\in I\times [m]\mid i\in I_j\}$
	and, for each $i\in I$, let $J_i = \{ j\in [m]\mid i\in I_j\}$.
	Note that $I_j$ are the fibers of $\Pi$ when we fix the second coordinate, while $J_i$
	are the fibers of $\Pi$ when we fix the first coordinate.
	From the data structure we can compute $\Pi$ and thus also the fibers $J_i$ for all $i$.
	We have
	\[
		\sum_{i\in I} |J_i| ~=~ |\Pi| ~=~ \sum_{j\in [m]} |I_j| ~=~ \sum_{j\in [m]} O(2^d B(n,d))
		~=~ O(2^d m \cdot B(n,d)).
	\]
	
	For each $i\in I$ we define the rational function
	\[
		\rho_i(x) ~=~ \sum_{p\in P_i} \frac{1}{x+\omega(p)}\, .
	\]
	For each $i\in I$, 	
	we use Lemma~\ref{lem:multipoint_evaluation} to evaluate $\rho_i(\delta_j)$
	for each $j\in J_i$. This means that, for each $i\in I$,
	we spend 
	\[
		O((|P_i| + |J_i|) \log^2 |P_i| \loglog |P_i|) ~=~ O((|P_i| + |J_i|) \log^2 n \loglog n) 
	\]
	time.
	For all $i\in I$ together we spend a total time of
	\begin{align*}
		\sum_{i\in I} O(|P_i| + |J_i|) \log^2 n \loglog n) ~&=~
		\left( \sum_{i\in I} |P_i| + \sum_{i\in I} |J_i| \right) O(\log^2 n \loglog n ) \\
		&=~
		\left( O(nd\cdot B(n,d)) + O(2^d m \cdot B(n,d)) \right) O(\log^2 n \loglog n )\\
		&=~ O((n+m)2^d B(n,d) \log^2 n \loglog n ).
	\end{align*}
	For each query $(R_j,\delta_j)$, since $P\cap R_j = \bigsqcup_{i\in I_j} P_i$,
	we have
	\[
		\ISW(R_j,\delta_j) ~=~ \sum_{p\in P\cap R_j}~ \frac{1}{\delta_j+\omega(p)}
						~=~ \sum_{i\in I_j} \sum_{p\in P_i}~ \frac{1}{\delta_j+\omega(p)} 
						~=~ \sum_{i\in I_j}~ \rho_i(\delta_j).
	\]
	Since the values $\rho_i(\delta_j)$ have been already computed for all $(i,j)\in \Pi$,
	the computation of $\ISW(R_j,\delta_j)$ for all $j$ takes additional 
	$O(|\Pi|)=O(2^d m \cdot B(n,d))$ time.
\end{proof}

\subsection{Across a separation}
Let $G$ be a graph with $n$ vertices and let $A,B,S$ be a separation of $G$.
Let $k$ be the size of the separator; thus $k=|S|$.
Our objective here is to compute $\IGL(A,B)$ efficiently. 
For the rest of the section, we assume that
$G$ and the partition $A,B,S$ are fixed.

We fix an order on the vertices of the separator $S$ and denote them by $s_1,\dots,s_k$.
For each $b\in B$, we want to partition $A$ into groups depending on which vertices of $S$
belong to a shortest path from $b$. To obtain a partition, we have to be careful
to assign each vertex of $A$ to one group. For this we use the minimum index
of the vertex in the separator that is contained in some shortest path. 
More precisely, for each $b\in B$ and each $i\in [k]$
we define the set
\begin{align*}
	A(b,i) ~:=~ \{ a\in A\mid & d_G(a,b) < d_G(a,s_j) + d_G(s_j,b) ~~~ \forall j<i,\\
							 & d_G(a,b) = d_G(a,s_i) + d_G(s_i,b), \\
							 & d_G(a,b) \le d_G(a,s_j) + d_G(s_j,b) ~~~ \forall j>i\}.	
\end{align*}
It is easy to see that the sets $A(b,i)$, $i\in [k]$, form a partition of $A$ for each $b\in B$.

Consider some index $i\in [k]$. 
Define the map $\phi^{(i)}=(\phi^{(i)}_j)_{j\in [k]}$ by
\begin{align*}
	\phi^{(i)}=(\phi^{(i)}_j)_{j\in [k]}\colon  A &\rightarrow \RR^k\\
						a &\mapsto (d_G(a,s_i)-d_G(a,s_j))_{j\in [k]}
\end{align*}
Note that the $i$-th coordinate of $\phi^{(i)}(a)$ is always $0$, and thus it does
not provide any information. 
We keep it to maintain slightly simpler notation.

For each vertex $b\in B$, we define the box 
$R^{(i)}(b)=I^{(i)}_1(b)\times\dots\times I^{(i)}_k(b)\subset \RR^k$,
where $I^{(i)}_j(b)$ is 
\[
	I^{(i)}_j(b) ~=~ \begin{cases}
			\bigl( -\infty,\, d_G(b,s_j)-d_G(b,s_i) \bigr) & \text{if $j<i$,}\\
			\RR & \text{if $j=i$,}\\
			\bigl( -\infty,\, d_G(b,s_j)-d_G(b,s_i) \bigr] & \text{if $j>i$.}
		\end{cases}
\]
 
\begin{lemma}
	A vertex $a\in A$ belongs to $A(b,i)$ if and only if $\phi^{(i)}(a)\in R^{(i)}(b)$.
\end{lemma}
\begin{proof}
	Vertex $a$ belongs to $A(b,i)$ if and only if
	\begin{align*}
		d_G(a,b) ~&<~ d_G(a,s_j) + d_G(s_j,b) ~~~ \forall j<i,\\
		d_G(a,b) ~&=~ d_G(a,s_i) + d_G(s_i,b), \\
		d_G(a,b) ~&\le~ d_G(a,s_j) + d_G(s_j,b) ~~~ \forall j>i.	
	\end{align*}
	Since each path from $a$ to $b$ passes through some vertex of $S$,
	this can be rewritten as 
	\begin{align*}
		d_G(a,s_i) + d_G(s_i,b) ~&<~ d_G(a,s_j) + d_G(s_j,b) ~~~ \forall j<i,\\
		d_G(a,s_i) + d_G(s_i,b) ~&\le~ d_G(a,s_j) + d_G(s_j,b) ~~~ \forall j>i.	
	\end{align*}
	Rearranging terms this is equivalent to 
	\begin{align*}
		\phi^{(i)}_j(a) ~&=~ d_G(a,s_i) - d_G(a,s_j) ~<~  d_G(s_j,b) -  d_G(s_i,b)~~~ \forall j<i,\\
		\phi^{(i)}_j(a) ~&=~ d_G(a,s_i) - d_G(a,s_j) ~\le~ d_G(s_j,b) -  d_G(s_i,b) ~~~ \forall j>i.	
	\end{align*}
	This last condition is precisely the condition for $\phi^{(i)}(a)\in R^{(i)}(b)$.
\end{proof}

\begin{lemma}
\label{lem:separation}
	Given a graph $G$ with $n$ vertices and $m$ edges,
	and a separation $A,B,S$ of $G$, where $S$ is the separator and has size $|S|=k\ge 2$, 
	we can compute $\IGL(A,B)$ in $O(km+n2^kk\cdot B(n,k-1)\cdot \log^2 n \loglog n)$ time.
\end{lemma}
\begin{proof}
	Let $s_1,\dots,s_k$ be an enumeration of the vertices of $S$.
	We compute shortest path trees from each vertex $s_i$ of $S$. This takes $O(k(m+n\log n))$ time.
	After this, each distance $d_G(s_i,v)$, where $s_i\in S$ and $v\in V(G)$, 
	is available in constant time.
	
	Consider a fixed index $i\in [k]$.
	We compute the set of points $P^{(i)}=\phi^{(i)}(A)$ and assign to point $\phi^{(i)}(a)$
	weight $\omega(\phi^{(i)}(a))= d_G(s_i,a)$. 
	For each $b\in B$ we compute the description of the box $R^{(i)}(b)$.
	This takes $O(kn)$ time, for a single $i\in [k]$.

	Since for each $b\in B$ and each $a\in A(b,i)$ 
	the shortest path from $b$ to $a$ goes through $s_i$, we have
	\begin{align*}
		\IGL(b,A(b,i))~&=~ \sum_{a\in A(b,i)} \frac{1}{d_G(a,b)} 
				~=~ \sum_{a\in A(b,i)} \frac{1}{d_G(a,s_i)+d_G(s_i,b)}\\
				~&=~ \sum_{a\in A(b,i)} \frac{1}{\omega(\phi^{(i)}(a)) + d_G(s_i,b)}
				~=~ \ISW(R^{(i)}(b),d_G(s_i,b)).
	\end{align*}	
	Recall that $P^{(i)}$ is effectively a point set in $k-1$ dimensions because
	the $i$-th coordinate is identically $0$.
	Using Theorem~\ref{thm:ISW} for $P^{(i)}$ we can compute 
	$\IGL(b,A(b,i))=\ISW(R^{(i)}(b),d_G(s_i,b))$
	for all points $b\in B$ together in 
	\[
		O((|A|+|B|)2^k\cdot B(n,k-1)\cdot \log^2 n \loglog n)~=~ 
		O(n2^k\cdot B(n,k-1)\cdot \log^2 n \loglog n)
	\]
	time.
	We repeat the procedure for each $i\in [k]$, which adds a multiplicative
	factor of $O(k)$ to the running time.
	After this, we have the values $\IGL(b,A(b,i))$
	for all $i\in [k]$ and $b\in B$.
	
	Finally we use that, for each $b\in B$, the sets
	$A(b,i)$, $i\in [k]$, form a partition of $A$ to obtain
	\[
		\IGL(A,B) ~=~ \sum_{b\in B} \sum_{a\in A} \frac{1}{d_G(a,b)} 
				 ~=~ \sum_{b\in B} \sum_{i\in [k]} \sum_{a\in A(b,i)}\frac{1}{d_G(a,b)} 
				 ~=~ \sum_{b\in B} \sum_{i\in [k]} \IGL(b, A(b,i)).
	\]
	Since the $O(k|B|)=O(kn)$ values used in the last expression are already available,
	the result follows.
\end{proof}

Note that in the proof of the previous Lemma we can actually compute
the values $\IGL(b,A)$ for all $b\in B$ within the same time bound.

\subsection{Final algorithm}
We use the recursive approach of Cabello and Knauer~\cite{CabelloK09}.
The idea is that graphs of small treewidth have small, balanced separators
with the property that adding edges between all pairs of vertices of the separator
does not increase the treewidth.
Given such a separation $A,B,S$ of the graph
we use Lemma~\ref{lem:separation} to compute $\IGL(A,B)$, 
we also compute $\IGL(S,V(G)\setminus S)$ using shortest paths from $S$, and
then solve recursively the problems for $A$ and $B$. For the recursive problems,
we add edges between the vertices of the separator $S$ whose length
is equal to the length of the shortest path in $G$. Thus, the distances
between vertices in $A$ and $B$ are correct in the recursive calls.
Also, adding those edges does not increase the treewidth of the graphs
used in the recursion. A bit of care is needed to avoid that
distances between vertices of $S$ are counted more than once;
it is easy to handle this by subtracting the terms that are counted twice.

Now we have two regimes depending on whether we want
to assume that the treewidth is constant, as done in~\cite{CabelloK09},
or whether we want to consider the treewidth a parameter, 
as done in~\cite{BringmannHM20}.
This difference affects the time to find a tree decomposition and
a balanced separator.
In both cases we use that an $n$-vertex graph with treewidth $k$
has $O(kn)$ edges~\cite{b-pkagb-98}. 

\begin{theorem}
\label{thm:treewidth1}
	Let $k\ge 2$ be an integer constant.
	For graphs $G$ with $n$ vertices and treewidth at most $k$,
	we can compute $\IGL(G)$ in $O(n \log^{k+2}n \loglog n)$ time.
\end{theorem}
\begin{proof}
	We obtain a tree decomposition of width $k$ in $O(n)$ time using
	the algorithm of Bodlaender~\cite{b-ltaft-96}, 
	which takes $O(2^{O(k^3 \log k)} n)=O(n)$ time, when $k$ is a constant.
	From the tree decomposition we can obtain a separation $A,B,S$,
	such that $A$ and $B$ contain $\Theta(n)$ vertices each, the separator
	$S$ has size $k$, and adding edges between the vertices of $S$ does not
	increase the treewidth; see~\cite[Lemma 3]{CabelloK09}.
	We use Lemma~\ref{lem:separation} to compute $\IGL(A,B)$.
	Using the estimate of Lemma~\ref{lem:B(n,d)} and that the number
	of edges is $m=O(kn)=O(n)$, we spend 
	\begin{align*}
		O(km+n2^kk\cdot B(n,k-1)\cdot \log^2 n \loglog n) ~&=~
		O(n \cdot \log^{k-1}n \cdot \log^2 n \loglog n) \\
		&=~
		O(n \log^{k+1}n\loglog n)
	\end{align*}
	time to compute $\IGL(A,B)$. 
	Because of the divide-and-conquer approach and because each $A$ and $B$ have
	a constant fraction of the vertices, 
	the recursion increases the running time by another logarithmic factor.
\end{proof}

\begin{theorem}
\label{thm:treewidth}
	For graphs $G$ with $n$ vertices and treewidth at most $k$,
	we can compute $\IGL(G)$ in $n^{1+\eps} 2^{O_\eps(k)}$ time, for any $\eps>0$.
\end{theorem}
\begin{proof}
	Bodlaender et al.~\cite{BodlaenderDDFLP16} give an algorithm that,
	for graphs of treewidth at most $k$, finds
	a tree decomposition of width $3k+4$ in $2^{O(k)}n\log n$ time.
	Given such a decomposition, we can obtain a separation $A,B,S$,
	such that $A$ and $B$ contain $\Theta(n-k)$ 
	vertices each, the separator
	$S$ has size $k'=(3k+4)+1=O(k)$, and adding edges between the vertices of $S$ does not
	increase the treewidth; see for example~\cite[Theorem 19]{b-pkagb-98}.
	We use Lemma~\ref{lem:separation} to compute $\IGL(A,B)$.
	Using the estimate of Lemma~\ref{lem:B(n,d)} and that the number
	of edges is $m=O(kn)$, we spend 
	\begin{align*}
		O(k'm+n2^{k'}k'\cdot B(n,k'-1)\cdot \log^2 n \loglog n) ~&=~
		O(k^2n + n2^{O(k)} \cdot n^{\eps}2^{O_\eps(k')} \cdot \log^2 n \loglog n) \\
		&=~
		n^{1+\eps} 2^{O_\eps(k)}
	\end{align*}
	time to compute $\IGL(A,B)$, for each $\eps>0$.
	The recursive calls add another logarithmic factor, which is absorbed 
	by the polynomial term $n^{1+\eps}$. (Actually, the logarithmic factor does not appear
	because the exponent of the polynomial is a constant strictly larger than $1$
	and that term dominates the recursive formulation of the running time.)
\end{proof}


\section{Planar graphs}
\label{sec:planar}

First we provide some tools for planar graphs, including $r$-divisions,
duality and Voronoi diagrams. 
Then we introduce a data structure to compute inverse shifted queries; 
this is our essential, new contribution. 
Finally we explain how to combine the data structure with 
pieces of an $r$-division, and how the whole algorithm works.

In the following we assume that $G=(V,E)$ is a fixed planar (or plane) graph 
and often drop the dependency on $G$ in the notation.

\subsection{Preliminaries on planar graphs}
A planar graph with $n$ vertices has $O(n)$ edges.
An embedding of a planar graph in the plane can be computed in linear time.
A \DEF{plane graph} is a planar graph together with a fixed embedding.
The embedding is usually described combinatorially using rotation systems:
it specifies for each vertex the circular ordering (clockwise, say)
of the edges incident to the vertex.
Adding edges of sufficiently large length to $G$, so that the distances
between vertices are not changed, we can assume that $G$ is triangulated, that is,
all faces of $G$ are triangles. Such addition of edges, without introducing
parallel edges, can be done in linear time.
Henceforth, we will assume that the planar graph $G$ is triangulated,
and thus $3$-connected. (When $n=3$, the graph is not $3$-connected, but the
case $n=O(1)$ is trivial.)

Let $G$ be a plane graph. A \DEF{piece} of $G$ is a subgraph of $G$ 
without isolated vertices. 
We assume in $P$ the embedding inherited from $G$.
The \DEF{boundary} of $P$, denoted by $\partial P$, is the set of vertices
of $P$ that are incident to some edge of $E\setminus E(P)$.
The \DEF{interior} of $P$ is $V(P)\setminus \partial P$.
Each path from a vertex of $P$ to a vertex of $V\setminus V(P)$ 
passes through some vertex of the boundary, $\partial P$.
Thus, $V\setminus V(P)$, $V(P)\setminus \partial P$, $\partial P$ 
is a separation of $G$. (It may be very unbalanced.)
A \DEF{hole} in a piece $P$ is a face of $P$ that is not a face of $G$.
Note that each boundary vertex of a piece $P$ must be incident
to some hole.

For any positive integer $r$, an \DEF{$r$-division with few holes} of an $n$-vertex
plane graph $G$ is a collection $\PP=\{ P_1,\dots,P_k \}$ of pieces of $G$ with the following properties:
\begin{itemize}
	\item each edge of $G$ belongs to precisely one piece $P_i\in \PP$;
	\item each piece $P_i\in \PP$ has at most $r$ vertices, $O(\sqrt{r})$ boundary vertices,
		and $O(1)$ holes;
	\item $k=|\PP|=O(n/r)$, that is, there are $O(n/r)$ pieces in $\PP$.
\end{itemize}
The first item implies that, if a vertex $v\in \partial P_i\in\PP$, then
$v$ belongs to the boundary of each piece that contains it.
Indeed, if such vertex $v$ belongs to two distinct pieces $P_i,P_j\in \PP$,
then $v$ is incident to an edge of $P_i$ and an edge of $P_j$ (pieces
do not have isolated vertices), and thus
belongs to the boundary of $P_j$ and the boundary of $P_i$. 
As a consequence, the interiors of distinct pieces are disjoint.
In particular, the sets 
$V(P_1)\setminus \partial P_1$, $\dots$, $V(P_k)\setminus \partial P_k$
and $\bigcup_{i\in [k]} \partial P_i$,
form a pairwise disjoint partition of $V$.

Klein, Mozes and Sommer~\cite{KleinMS13} give an algorithm to compute in $O(n)$ time 
an $r$-division with few holes, given a plane, triangulated graph $G$ and a value $r$.
A simpler algorithm with running time $O(n \log n)$ was already known 
before~\cite{FakcharoenpholR06,KleinS98}.

\begin{figure}[htb]
	\centering
	\includegraphics[page=1]{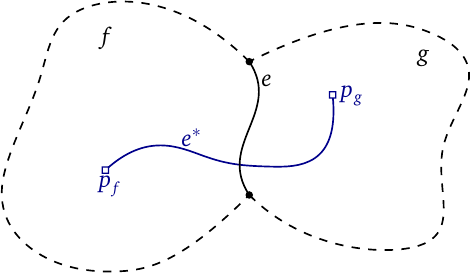}
	\caption{Embedding of the dual graph (in blue).}
	\label{fig:dual}
\end{figure}

Let $G$ be a plane graph and denote by $F(G)$ its set of faces.
The \DEF{dual} (multi)graph of $G$, denoted as $G^*$, has vertex set $V(G^*)=F(G)$
and edge set 
\[
	E(G^*)=\{ e^*=fg\mid f,g\in F(G), \text{ faces $f$ and $g$ of $G$ share an edge $e\in E(G)$}\}.
\]
We use $e^*$ to denote the edge of $G^*$ dual to the edge $e\in E(G)$.
The embedding of $G$ defines naturally an embedding of the dual graph $G^*$, as follows.
Each face $f\in F(G)$ is represented by a point $p_f$ chosen inside the face $f$,
and the dual edge $fg=e^*$ is drawn as a curve with endpoints $p_f$ and $p_g$
that crosses the embedding of $G$ exactly once, namely in the interior of the drawing of $e$.
See Figure~\ref{fig:dual}.
In general plane graphs, the dual graph is a multigraph because two faces
can share several edges. To keep the notation simpler, we will talk about
$G^*$ as a simple graph; adapting it to multigraphs is straightforward.

Let $\gamma$ be a cycle in the dual graph $G^*$ and let $s$ 
be a vertex of $G$.
Let $A$ and $B$ be the two connected components of 
$\RR^2\setminus \gamma$, and assume without loss of generality that $s$ belongs to $A$. 
We define the \DEF{exterior} of $\gamma$ with respect to $s$ as 
$\exterior_G(\gamma,s)=V(G)\cap B$. See Figure~\ref{fig:exterior}.
Thus, the exterior of $\gamma$ (with respect to $s$) is the set of vertices 
of $G$ such that any path connecting them to $s$ must intersect $\gamma$.
While in the plane there is a usual meaning of exterior and interior of a closed
Jordan curve, our definition is the natural parallel to the concept in the $2$-dimensional
sphere and will be more convenient in our application.

\begin{figure}[htb]
	\centering
	\includegraphics[page=2,width=.95\textwidth]{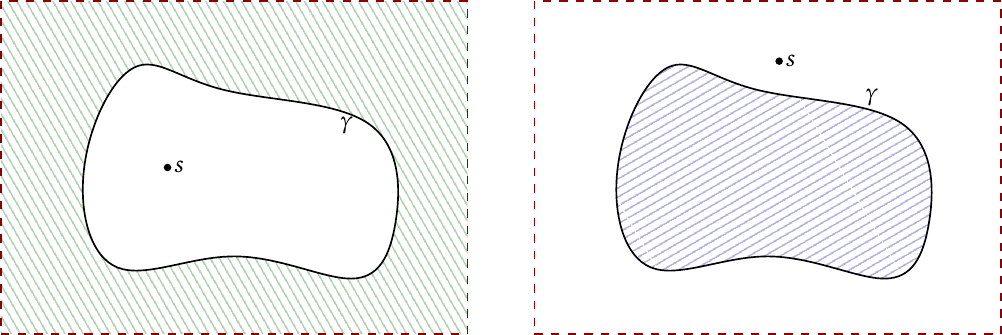}
	\caption{Two examples showing the region that would contain the
		vertices $\exterior_G(\gamma,s)$. In the left, this would be the unbounded
		region of $\RR^2\setminus \gamma$, while in the right this would be
		the bounded region. (The red dashed line shows the clipping of $\RR^2$ for the drawing.)}
	\label{fig:exterior}
\end{figure}

\subsection{Voronoi diagrams in plane graphs}
The following material is taken from Gawrychowski et al.~\cite{GawrychowskiKMS21}, where there is 
a comprehensive treatment. See also Cabello~\cite{Cabello2019} for the
initial treatment, where only a particular case is discussed.

Let $G$ be a plane graph with (abstract) edge lengths
and let $P$ be a piece of $G$. 
(In fact, the graph $G$ is irrelevant for the forthcoming discussion; 
we use $P$ and $G$ to be consistent with the use later.) 
Let $S$ be a subset of vertices of $P$, the \DEF{sites},
and let $\delta\colon S\rightarrow \RR_{\ge 0}$
assign non-negative weights to the sites. Usually we only talk about
$S$ and treat $\delta$ as implicitly given with $S$.
The (additively-weighted) \DEF{Voronoi diagram} of $P$ with respect to $S$ (and $\delta$) 
is the family of subsets of vertices
\[
	\cell_P(s,S)~=~  \{ v\in V(P)\mid \forall s'\in S:~\delta(s)+d_P(s,v) \le \delta(s')+d_P(s',v)\}
	~~~~ \text{ for all $s\in S$.} 
\]
Here, $\cell_P(s,S)$ is the \DEF{Voronoi cell} for $s$. Note that we are using the 
distances in $P$; the graph $G$ is irrelevant.
A Voronoi cell may be the empty set.

It is convenient that the cells of the Voronoi diagram form a partition of $V(P)$.
For this, we can break ties in several consistent ways. For example, we may index
the sites and say that, in case of ties, the vertex gets assigned 
to the closest site with smallest index.
(There are other valid options.) Formally, we would enumerate
$S$ as $S=\{ s_1,\dots,s_k\}$ and define for each site $s_i$ the Voronoi cell
\begin{align*}
	\cell_P(s_i,S)~=~  \{ v\in V(P)\mid\, 
							& \forall j<i:~\delta(s_i)+d_P(s_i,v) < \delta(s_j)+d_P(s_j,v), \\
							& \forall j>i:~\delta(s_i)+d_P(s_i,v) \le \delta(s_j)+d_P(s_j,v)\}.
\end{align*}

It is easy to see that the Voronoi cell for $s\in S$ is \DEF{star-shaped} from $s$:
for each vertex $v\in \cell_P(s,S)$, the vertices in any shortest path in $P$ from $s$ to $v$
are also contained in $\cell_P(s,S)$. 
In particular, $\cell_P(s,S)$ forms a \emph{connected} subgraph of $P$.
In the forthcoming claims we will assume (for simplicity) that each Voronoi cell is nonempty.
Empty Voronoi cells have to be treated as a special, degenerate case.

Since each Voronoi cell is connected, each (non-empty) Voronoi cell 
can be described by a collection of cycles in the dual graph $P^*$, as follows. 
Here it is important that $P$ is a \emph{plane} graph, not just planar.
For each site $s\in S$, there is a family of cycles $H_P(s,S)$
in the dual graph $P^*$ such that 
\[
	\cell_P(s,S)~=~ V(P) \setminus \left( \bigsqcup_{\gamma\in H_P(s,S)} \exterior_P(\gamma,s)\right).
\]
See Figure~\ref{fig:holes}.
Recall that $\sqcup$ denotes that the sets in 
$\{ \exterior_P(\gamma,s)\mid \gamma\in H_P(s,S)\}$ are pairwise disjoint (for each single $s$).
We call the family of cycles $H_P(s,S)$ the \DEF{dual description} of $\cell(s,S)$.
(We use $H_P(s,S)$ because of their role as ``holes''.)
The existence of such cycles follows from the following fact for planar graphs $P$:
the removal of $E'\subset E(P)$ in $P$ leaves exactly two connected components 
if and only if the dual edges $(E')^*=\{ e^*\mid e\in E'\}$ define a simple cycle in the dual
graph $P^*$. More precisely, a dual edge $e^*$ belongs to some 
cycle of $H_P(s,S)$ if and only if $e$ has one vertex in $\cell_P(s,S)$
and one vertex outside $\cell_P(s,S)$. Thus, each connected component of
$G-\cell_P(s,S)$ defines one cycle of $H_P(s,S)$.
The behavior is parallel to what occurs with Voronoi diagrams in the plane
for some non-Euclidean metrics; think for example of a Voronoi diagram
in triangulated terrain, where some Voronoi cells may have holes. (Voronoi
cells are not necessarily simply connected in such an scenario.)
A Voronoi cell in such a Voronoi diagram is described by a collection of curves,
where each curve describes a ``hole'', that is, a part that does not belong
to the Voronoi cell.

\begin{figure}[htb]
	\centering
	\includegraphics[page=3,width=.95\textwidth]{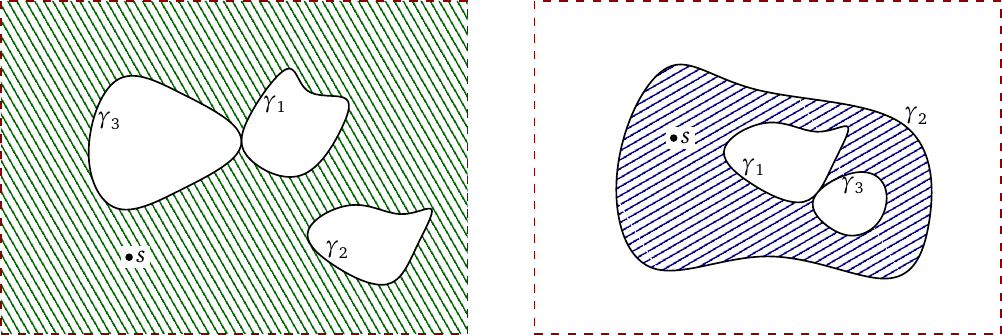}
	\caption{Two examples showing how the cycles of $H_P(s,S)$ define $\cell_P(s,S)$.
		In both examples $H_P(s,S)=\{ \gamma_1,\gamma_2,\gamma_3\}$ and the vertices
		of $\cell_P(s,S)$ would be contained in the marked region.
		(The red dashed line shows the clipping of $\RR^2$ for the drawing.)}
	\label{fig:holes}
\end{figure}

A \DEF{bisector} is the cycle in the dual graph separating the two (non-empty) Voronoi cells of
a two-site Voronoi diagram. That is, given two additively weighted sites $s$ and $t$,
we observe that $H_P(s,\{s,t\})$ consists of a single cycle because both
$\cell_P(s,\{s,t\})$ and $\cell_P(t,\{s,t\})$ are connected. We denote such
bisector as $\beta(s,t)$ and note that it consists of 
the set of edges dual to
\[
	\left\{ uv\in E(P)\text{ such that }  |\{u,v\}\cap \cell_P(s,\{s,t\})|=1 \right\}.
\]

When we have multiple sites, the cycles 
in the dual representation $H_P(s,S)$ describing the boundary 
of $\cell(s,S)$ are made of portions of bisectors. (Portions in general,
but some of them could be a whole bisector.) 
This is the case because, for each dual edge $e^*$ belonging 
to some cycle in $H_P(s,S)$, one vertex of $e$ belongs
to $\cell_P(s,S)$ while the other vertex of $e$ lies in $\cell_P(t,S)$ 
for some $t\in S \setminus \{ s \}$. 
Each cycle in $H_P(s,S)$ can be described by a sequence of \DEF{patches},
where each patch is a maximal portion of a bisector $\beta(s,t)$ for some 
site $t\in S\setminus \{s\}$. Such a patch
is described by a tuple $(s,t,e_1,e_2)$, 
where $t\in S\setminus \{ s\}$ and 
where $e_1,e_2\in \beta(s,t)$ are the first and last edge, respectively, 
in the portion of the bisector (in some prescribed direction for $\beta(s,t)$).
Thus, the description of one patch requires $O(1)$ space.

The combinatorial complexity of $H_P(s,S)$ or $\cell_P(s,S)$ is the number of patches
required for all cycles in $H_P(s,S)$.
The \DEF{combinatorial complexity of the Voronoi diagram} defined by the sites $S$ 
is the sum of the combinatorial complexities of $H_P(s,S)$ over all $s\in S$.
Using properties of planar graphs one can see that the combinatorial 
complexity of Voronoi diagram of $S$ is $O(|S|)$. Here it is important
to note that the size of the piece is not relevant.
By the computation of the Voronoi diagram we mean the computation of the
cycles $H_P(s,S)$ for all $s\in S$, where each cycle of each $H_P(s,S)$
is given as a sequence of patches. 
(In the description of the Voronoi diagram one usually considers also the adjacency
relations between cells; we will not use such information explicitly in our algorithm, but it
is implicit in the description of the patches because they contain information about 
which bisectors are being used.)

Cabello provided a randomized algorithm to compute Voronoi diagrams efficiently,
after an expensive preprocessing, in a particular case when the sites are cofacial. 
Gawrychowski et al.~improved the result in several aspects, allowing more
general families of sites, making the preprocessing faster and the construction deterministic,
as follows.

\begin{theorem}[Gawrychowski et al.~\cite{GawrychowskiKMS21}]
\label{thm:VDplanar}
Let $P$ be a plane graph with $r$ vertices, let 
$S$ be a set of $b$ vertices in $P$, 
assume that the vertices of $S$ can be covered with $O(1)$ faces of $P$,
and such covering with $O(1)$ faces is given.
After $\tilde O(rb^2)$ preprocessing time we can handle the following queries:
for any given weights $\delta(s)$ assigned to each $s\in S$ at query time,
the additively-weighted Voronoi diagram of $P$ with respect to $S$ can be computed
in $\tilde O(b)$ time.
\end{theorem}

We will be using the result for pieces $P$ in an $r$-division of a plane graph $G$.
In such case we will have $b=O(\sqrt{r})$, the preprocessing
will take $\tilde O(r^2)$ time, 
and the computation of each Voronoi diagram takes $\tilde O(\sqrt{r})$ time.

\subsection{Inverse shifted queries}
\label{sec:inverseshiftedplanar}
Let $P$ be a piece with $r$ vertices in a plane graph $G$, and let $s$ be a fixed vertex of $P$.
Assume that each vertex $v$ of $P$ has a prescribed non-negative weight $\omega(v)$.
As a warm up for the forthcoming technique, we provide the following result.

\begin{lemma}
\label{le:wholepiece}
	For $t$ given positive values $\delta_1,\dots,\delta_t$,
	we can compute the values
	\[
		\sum_{v\in V(P)}\,\frac{1}{\delta_i+\omega(v)}
		~~~~ \text{for all $i\in[t]$}.
	\]
	in $O((r+t) \log^2 r \loglog r )$ time.
\end{lemma}
\begin{proof}
	Consider the rational function 
	$\rho(x)=\sum_{v\in V(P)}\,\frac{1}{x+\omega(v)}$.
	We evaluate $\rho(x)$ at $x=\delta_i$ for all $i\in [t]$.
	Using Lemma~\ref{lem:multipoint_evaluation}, we can compute
	all these values in $O((r+t)\log^2 r \loglog r)$ time.
\end{proof}

In general we will be considering similar sums as in the previous result,
but for subsets of vertices defined by cycles in the dual graph.
Recall that for a cycle $\gamma$ in the dual graph $P^*$ and
a vertex $s$ in $P$, 
$\exterior_P(\gamma,s)$ is the set of vertices of $P$ that are on the opposite
side of $\gamma$ than $s$.

Given a cycle $\gamma$ in $P^*$, 
and a positive weight $\delta$ (the \emph{shift}),
we define the \DEF{inverse shifted-weight} (ISW) values
\begin{align*}
	\ISW_{P,s}(\gamma,\delta) ~&=~ \sum_{v\in \exterior_P(\gamma,s)} \,\frac{1}{\delta+\omega(v)}.
\end{align*}
Note that we are adding the inverse of $\delta+\omega(v)$, and we are only taking
into account the vertices of $P$ outside $\gamma$.
Since $\omega(v)$ is non-negative an $\delta$ is positive,
the denominator is always positive, and thus the fractions are well-defined.

We are interested in the efficient computation of the values $\ISW_{P,s}$ 
for several pairs of the form $(\gamma,\delta)$. More precisely, we will handle
the batched version: computing $\ISW_{P,s}(\gamma,\delta)$ for several
given pairs $(\gamma,\delta)$, where $P$ and $s$ remain constant. 

In our application we will consider only cycles $\gamma$ 
with a particular structure that we explain now.
A cycle $\gamma$ in $P^*$ is \DEF{$s$-star-shaped} if the following holds:
for each $u,v\in V(P)$, if $u$ belongs to the shortest path in $P$ 
from $s$ to $v$ and $u\in \exterior_P(\gamma,s)$, then
$v\in \exterior_P(\gamma,s)$. (The intuition is that the interior of $\gamma$ 
with respect to $s$ has to be star-shaped from $s$.)

Let $\Xi(s)$ be the family of cycles in $P^*$ that are $s$-star-shaped. 
We will be interested in $\ISW_{P,s}$ queries for some cycles in $\Xi(s)$.
We will assume that the cycles $\gamma$ in $\Xi(s)$ are oriented in such a way that $s$
lies in the region of $\RR^2\setminus \gamma$ to the right of $\gamma$, as we walk
along $\gamma$. Thus, if $s$ is in the bounded region of $\RR^2\setminus \gamma$,
then $\gamma$ is oriented clockwise, and otherwise $\gamma$ is oriented counterclockwise.

In the following we consider the dual graph $P^*$ as an oriented graph that contains both orientations
of each edge. The word \emph{arc} is used for \emph{oriented} edges.
For a primal arc $\dart uv$ with face $a$ to the left and face $b$ to the right, 
we define its dual arc $(\dart uv)^*=\dart ab$. Thus, the dual arc crosses the primal
arc from left to right. See Figure~\ref{fig:dualarc}.
For simplicity we assume that each arc has different endpoints. It is easy
to enforce this or modify the discussion to the general case.
Recall that $r$ is the number of vertices in $P$.

\begin{figure}[htb]
	\centering
	\includegraphics[page=4,width=.95\textwidth]{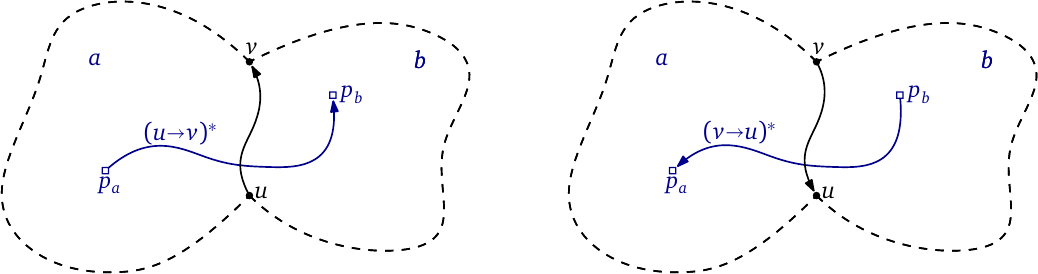}
	\caption{Duality for arcs. In the left we have $(\dart uv)^*$ and in the right
		$(\dart vu)^*$.}
	\label{fig:dualarc}
\end{figure}

\begin{lemma}
\label{lem:planar_weights}
	We can associate to each dual arc $\dart ab$ of $P^*$ a set $U(\dart ab)$ of vertices
	of $P$ with the following properties: 
	\begin{itemize}
	\item For each oriented cycle $\gamma\in\Xi(s)$,
	\[
		\exterior_P(\gamma,s) ~=~ 
			\bigsqcup_{\dart ab\in E(\gamma)} 
				U(\dart ab).
	\]
	\item The sets $U(\dart ab)$ 
	for all $\dart ab\in E(P^*)$ together can be computed in $O(r^2)$ time.
	\end{itemize}
\end{lemma}
\begin{proof}
	Fix a shortest path tree $T_s$ from $s$ in $P$
	and orient the edges as arcs away from the root $s$.
	For each arc $\dart uv$ of $T_s$, 
	let $T_{\dart uv}$ be the subtree of $T_s-uv$ that contains $v$.
	For each arc $\dart uv$ of $T_s$, we define $U((\dart uv)^*)=V(T_{\dart uv})$.
	For all other arcs $\dart ab$ of $P^*$ we define $U(\dart ab)=\emptyset$.
	See Figure~\ref{fig:setsU} for the definition and intuition in the forthcoming
	argument.
	
	\begin{figure}[tbh]
		\centering
		\includegraphics[page=5,width=.95\textwidth]{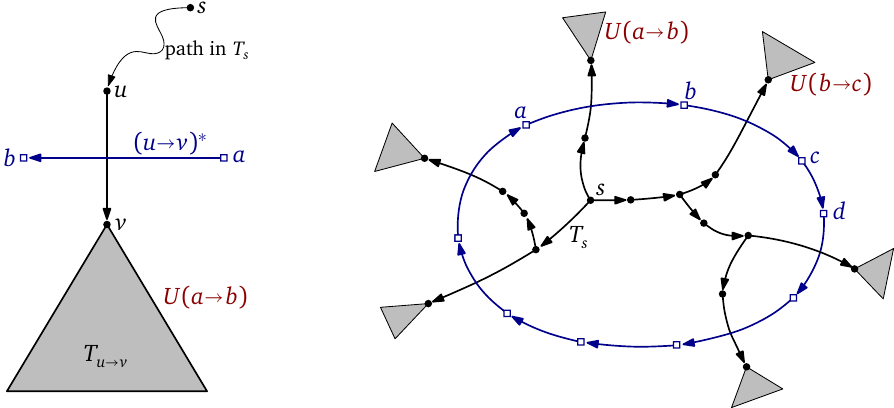}
		\caption{Left: The definition of $T_{\dart uv}$ and $U((\dart uv)^*)=U(\dart ab)$
			for an arc $\dart uv$ of $P$ and the dual arc $\dart ab$ in $P^*$.
			Right: The intuition in the construction. 
			The blue clockwise cycle $\gamma=\dart ab,\dart bc,\dots$ is in the dual graph $P^*$.		
			In this example $s$ is in the bounded region of $\RR^2\setminus \gamma$
			and $U(\dart{c}{d})=\emptyset$.}
		\label{fig:setsU}
	\end{figure}

	The description of the sets $U(\cdot)$ we gave readily leads to an
	algorithm to compute them with time complexity $O(r^2)$ because $P$ has $O(r)$ edges.
	
	Consider a cycle $\gamma\in \Xi(s)$, that is, $\gamma$ is $s$-star-shaped. 
	Then the edges of $T_s$ crossed by $\gamma$ are not in any ancestor-descendant relation.
	That is, in any $s$-to-leaf path in $T_s$ there is at most one edge whose
	dual edge appears in $\gamma$.
	Indeed, if there would be two edges $e_1$ and $e_2$ of $T_s$ crossed by $\gamma$, 
	then along the path in $T_s$ between $e_1$ and $e_2$ would get an alternation
	interior-exterior-interior with respect to $\gamma$ and $s$, 
	and we would get a contradiction with the property that $\gamma$ is $s$-star-shaped.
	
	Since the edges of $T_s$ crossed by $\gamma\in \Xi(s)$ are not in any ancestor-descendant relation,
	it follows from the construction that the sets $U(\dart ab)$ for all $\dart ab\in E(\gamma)$
	are pairwise disjoint. 
	
	Consider any vertex $v\in V(P)$.
	If $v\in \exterior_P(\gamma,s)$, then the path in $T_s$ from $s$
	to $v$ crosses $\gamma$ exactly once, and thus $v$ belongs to exactly one set 
	$U(\dart ab)$ with $\dart ab\in E(\gamma)$.
	If $v\notin \exterior_P(\gamma,s)$, then the path in $T_s$ from $s$
	to $v$ does not cross $\gamma$, and thus $v$ does not belong to any
	$U(\dart ab)$ with $\dart ab\in E(\gamma)$.
\end{proof}

Using ideas similar to those used in segment trees~\cite[Section 10.3]{bkos-08} we obtain the following.
	
\begin{lemma}
\label{lem:a_cycle}
	Assume that the sets $U(\dart ab)$ of Lemma~\ref{lem:planar_weights} are already computed.
	Let $\gamma$ be an oriented cycle in $\Xi(s)$.
	Assume that we are given $t$ pairs $(\pi_1,\delta_1),\dots, (\pi_t,\delta_t)$,
	such that each $\pi_i$ is a subpath of $\gamma$ specified by the starting and ending arc,
	and each $\delta_i>0$.
	In $O((r+t)\log^3 r \loglog r)$ time we can compute the values
	\[
		\partsum(\pi_i,\delta_i) ~:=~ \sum_{\dart ab\in E(\pi_i)} 
				\left(~\sum_{v\in U(\dart ab)}~\frac{1}{\delta_i+\omega(v)}~\right)
		~~~~ \text{for all $i\in[t]$}.
	\]
\end{lemma}
\begin{proof}
	Let $e_1,\dots ,e_m$ be the arcs along $\gamma$, and note that $m=O(r)$.  
	We make a rooted balanced binary search tree $T$ 
	with $m$ leaves such that the leaf storing $e_i$ has search key $i$. 
	A left-to-right traversal of $T$ traverses the leaves containing $e_1,\dots ,e_m$ in that order.
	See Figure~\ref{fig:treeT}.
	Each node $z$ of $T$ represents a subpath of $\gamma$, denoted by $\gamma[z]$,
	which is the concatenation of arcs stored in leaves under $z$.
	Alternatively, we can define $\gamma[z]$ recursively:
	for each leaf $z$ that stores $e_i$, we define $\gamma[z]=e_i$,
	and for an internal node $z$ with left child $z'$ 
	and right children $z''$ the path
	$\gamma[z]$ is the concatenation of $\gamma[z']$ and $\gamma[z'']$.
	The paths $\gamma[z]$ for all nodes $z$ of $T$ are called \emph{canonical subpaths} of $\gamma$.

	\begin{figure}[tbh]
		\centering
		\includegraphics[page=6,width=.95\textwidth]{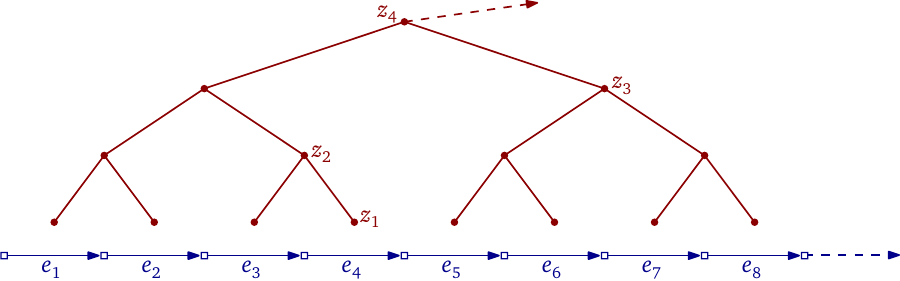}
		\caption{Portion of the tree $T$ for the part $e_1,e_2,\dots,e_8$ of $\gamma$.
			In this example, $\gamma[z_1]=e_4$, $\gamma[z_2]=e_3,e_4$, 
			$\gamma[z_3]=e_5,\dots,e_8$ and $\gamma[z_4]=e_1,\dots,e_8$.}
		\label{fig:treeT}
	\end{figure}
	
	For each node $z$ of $T$, let $r_z$ be $\sum_{e_i\in \gamma[z]} |U(e_i)|$.
	Note that for each leaf $z$ of $T$ that stores $e_i$ we have $r_z=|U(e_i)|$,
	while for each internal node $z$ of $T$ with children $z'$ and $z''$
	we have $r_z=r_{z'} + r_{z''}$.
	Since the sets $U(e_1),\dots,U(e_m)$ are pairwise disjoint 
	because of Lemma~\ref{lem:planar_weights} and canonical paths for nodes 
	at the same level of the tree $T$ are arc-disjoint,
	we have $r_z\le r$ for all $z$ and $\sum_{z\in V(T)} r_z= O(r \log m) = O(r\log r)$.
	
	Each subpath $\pi_i$ of $\gamma$ is the union of $O(\log m)=O(\log r)$ canonical subpaths.
	That is, for each subpath $\pi_i$ there is a subset $Z_i$ of $O(\log r)$ nodes in $T$
	such that $\pi_i$ is the concatenation of the canonical subpaths $\gamma[z]$, $z\in Z_i$.
	Furthermore, the nodes $Z_i$ describing those canonical subpaths can be identified in $O(\log r)$
	time with a bottom-up traversal from the first and the last arcs of $\pi_i$.
	(Recall that we assume that $\pi_i$ is specified by the starting and ending arc.)
	This property is parallel to the property often used for 
	segment trees~\cite[Lemma 10.10]{bkos-08}.
	We compute in $O(t\log r)$ time the sets $Z_1,\dots,Z_t$.
	For each node $z$ of $T$, let $I_z=\{ i\in [t]\mid z\in Z_i\}$.
	The non-empty sets $I_z$ for all nodes $z$ together can be computed also in $O(t\log r)$ time.
	Note that 
	\[	
		\sum_{z\in V(T)} |I_z| ~=~ \sum_{i\in [t]} |Z_i| ~=~ \sum_{i\in [t]} O(\log r) ~=~ O(t\log r).
	\]
	
	For each node $z$ of $T$ we define the rational function
	\[
		\rho_z(x) ~=~ \sum_{\dart ab\in E(\gamma[z])} 
			~~ \sum_{v\in U(\dart ab)} \,\frac{1}{x+\omega(v)}.
	\]
	Note that $\rho_z$ is the sum of 
	\[
		\sum_{e_i\in E(\gamma[z])} |U(e_i)| ~=~ r_z
	\]
	rational functions, each of bounded degree.	

	For each node $z$ of $T$ we compute $\rho_z(\delta_i)$
	for each $i\in I_z$; that is, we evaluate $\rho_z(x)$ at $x=\delta_i$ for all $i\in I_z$.
	Using Lemma~\ref{lem:multipoint_evaluation} this takes
	$O((|I_z|+ r_z )\log^2 r_z \loglog r_z)=O((|I_z|+ r_z )\log^2 r \loglog r)$ time
	for each node $z$ of $T$.
	Thus, for all nodes $z$ of $T$ and all $i\in I_z$ together, we compute 
	$\rho_z(\delta_i)$ in time
	\begin{align*}
		\sum_{z\in V(T)} O((|I_z|+ r_z )\log^2 r \loglog r) ~&=~
		O(\log^2 r \loglog r) \sum_{z\in V(T)} (|I_z|+ r_z ) \\ &=~
		O(\log^2 r \loglog r) \Bigl( O(t\log r) +  O(r \log r) \Bigr) \\ &=~
		O((r+t)\log^3 r \loglog r).
	\end{align*}
	
	Note that for each given pair $(\pi_i,\delta_i)$, 
	since $\pi_i$ is the concatenation of $\gamma[z]$, $z\in Z_i$, we have 
	\[
		\sum_{\dart ab\in E(\pi_i)} 
				\left(~\sum_{v\in U(\dart ab)}~\frac{1}{\delta_i+\omega(v)}~\right)
			~=~ \sum_{z\in Z_i} ~~ \sum_{\dart ab\in E(\gamma[z])} 
				~~ \sum_{v\in U(\dart ab)} \frac{1}{\delta_i+\omega(v)}
			~=~ \sum_{z\in Z_i} ~~ \rho_z(\delta_i).
	\]
	Since the values $\rho_z(\delta_i)$ are already available, for each $i\in [t]$
	we spend $O(|Z_i|)= O(\log r)$ additional time.	The result follows.
\end{proof}

Here we have the final result that we will use.

\begin{theorem}
\label{thm:datastructure2}
	Let $P$ be a piece with $r$ vertices in a plane graph $G$, and let $s$
	be a vertex of $P$.
	Let $\Gamma=\{ \gamma_1,\dots,\gamma_\ell \}$ be a family of oriented cycles in $\Xi(s)$.
	Assume that we are given $t$ pairs $(\pi_1,\delta_1),\dots, (\pi_t,\delta_t)$,
	where each $\delta_i>0$, each $\pi_i$ is an element of $\Xi(s)$,
	and each $\pi_i$ is given as a concatenation of $k_i$ subpaths
	of cycles from $\Gamma$. Set $k=\sum_i k_i$. 
	We can compute $\ISW_{P,s}(\pi_1,\delta_1),\dots \ISW_{P,s}(\pi_t,\delta_t)$
	in $O(r^2 + (k+\ell r) \log^3 r \loglog r )$ time.
\end{theorem}
\begin{proof}
	Let $\Pi_i$ be the family of subpaths of cycles from $\Gamma$ that are given
	to describe $\pi_i$. Let $\Pi=\cup_i \Pi_i$ be all the subpaths used
	in the description; we consider $\Pi$ as a multiset. For each $\varpi\in \Pi$, 
	let $i(\varpi)\in [t]$ be the index such that $\varpi\in \Pi_{i(\varpi)}$ and 
	let $j(\varpi)\in [\ell]$ be the index such that $\varpi$ is a subpath of $\gamma_{j(\varpi)}$.
	Let $\Gamma_j =\{ \varpi\in \Pi\mid j(\varpi)=j \}$, that is, the subpaths of $\gamma_j$
	that appear in the description of some $\pi_i$. Note that
	\[
		\sum_{j\in [\ell]} |\Gamma_j| ~=~ \sum_{i\in [t]} |\Pi_i| ~=~
		\sum_{i\in [t]} k_i ~=~ k.
	\]
		
	First we compute the sets $U(\dart ab)$ described in Lemma~\ref{lem:planar_weights}
	for all dual arcs $\dart ab$ in $P^*$. This takes $O(r^2)$ time.
	
	For each $j\in [\ell]$, we use Lemma~\ref{lem:a_cycle}
	with the cycle $\gamma_j$ and the pairs $(\varpi,\delta_{i(\varpi)})$, 
	for all $\varpi\in \Gamma_j$. This means that in time
	$O((r+|\Gamma_j|)\log^3 r \loglog r)$ we compute the sums
	\[
		\partsum(\varpi,\delta_{i(\varpi)}) ~=~\sum_{\dart ab\in E(\varpi)} 
				\left(~\sum_{v\in U(\dart ab)}\,\frac{1}{\delta_{i(\varpi)}+\omega(v)}~\right)
		~~~~ \text{for all $\varpi \in \Gamma_j$}.
	\]
	For all $j\in [\ell]$ together we spend
	\begin{align*}
		\sum_{j\in [\ell]} O((r+|\Gamma_j|)\log^3 r \loglog r) ~&=~ 
		O(\log^3 r \loglog r) \sum_{j\in [\ell]} (|\Gamma_j|+r) \\ &=~
		O( (k+\ell r) \log^3 r \loglog r )
	\end{align*}
	time and we obtain $\partsum(\varpi,\delta_{i(\varpi)})$ for all $\varpi\in \Pi$.
	This means that we obtain $\partsum(\varpi,\delta_i)$ for all $\varpi\in \Pi_i$
	and all $i\in [t]$.

	For each $(\pi_i,\delta_i)$, we use the properties of 
	$U(\dart ab)$ described in Lemma~\ref{lem:planar_weights},
	namely that $\exterior_P(\pi_i,s)$ is the \emph{disjoint} union of
	the sets $U(\dart ab)$ over all arcs $\dart ab \in E(\pi_i)$,
	to obtain
	\begin{align*}
		\ISW_{P,s}(\pi_i,\delta_i) ~&=~ 
				\sum_{v\in \exterior_P(\pi_i,s)} \,\frac{1}{\delta+\omega(v)}\\
				&=~\sum_{\dart ab\in E(\pi_i)} 
				\left(~\sum_{v\in U(\dart ab)}\,\frac{1}{\delta_i+\omega(v)}~\right)\\
				&=~ \sum_{\varpi\in \Pi_i}~~\sum_{\dart ab\in E(\varpi)} 
				\left(\sum_{v\in U(\dart ab)}\,\frac{1}{\delta_i+\omega(v)}\right)\\
				&=~ \sum_{\varpi\in \Pi_i} \partsum(\varpi,\delta_i).
	\end{align*}
	Thus, we can compute $\ISW_{P,s}(\pi_i,\delta_i)$ using an additional $O(|\Pi_i|)=O(k_i)$ time, 
	for each $i$.
\end{proof}

Lemma~\ref{le:wholepiece} and Theorem~\ref{thm:datastructure2}
will be combined to compute the sum of inverses over the vertices of a Voronoi cell
using simple inclusion-exclusion; see the end of the proof of Lemma~\ref{le:perpiece}.

\paragraph{Remark.} In our discussion we have assumed that we want to add the shifted-inverse
over vertices of $V(P)$ outside $\gamma$ with respect to a reference $s\in V(P)$. 
In fact, for any fixed $U\subset V(P)$,
we could also be considering the inverse-shifted queries 
\begin{align*}
	\ISW_{P,s}(\gamma,\delta) ~&=~ \sum_{v\in U\cap \exterior_P(\gamma,s)} \frac{1}{\delta+\omega(v)}.
\end{align*}
For example, if we would have red and blue vertices, we could consider sums over only
the blue vertices.

\paragraph{Remark.}
Similar results hold for arbitrary cycles, not just cycles from $\Xi(s)$,
with worse running time.
In Lemma~\ref{lem:a_cycle} we have used that the sets $U(\dart ab)$ are pairwise disjoint 
when the arcs $\dart ab$ come from a cycle of $\Xi(s)$. 
For arbitrary cycles one can use inclusion-exclusion;  
see for example Cabello~\cite[Section 3.1]{Cabello2019} for the relevant idea.
In such case, we can assume only that each $U(\dart ab)$ has $O(r)$ vertices
and, in the proof of Lemma~\ref{lem:a_cycle},
$\sum_{z\in V(T)} r_z= O(r^2)$.
The running time of Lemma~\ref{lem:a_cycle} for arbitrary cycles becomes 
$O((r^2+t\log r)\log^2 r \loglog r)$.
Similarly, Theorem~\ref{thm:datastructure2} holds for arbitrary cycles $\gamma_1,\dots,\gamma_\ell$,
but the running time becomes 
$O((k \log r+\ell r^2) \log^2 r \loglog r )$.

\subsection{Across a separation}
The next result handles the interaction between a piece $P$ of a plane graph $G$ 
and its complement. The boundary $\partial P$ of $P$ plays the role of the separation
because each path from $V(P)\setminus \partial P$ to $V(G)\setminus V(P)$ must
pass through $\partial P$.
We use Voronoi diagrams for planar graphs, 
as Cabello~\cite{Cabello2019} or Gawrychowski et al.~\cite{GawrychowskiKMS21},
to group the vertices of $V(P)$ depending on which boundary vertex is used
in the shortest path to them (for a fixed source).
To collect the relevant data from each Voronoi cell, we use the data structure
developed in Section~\ref{sec:inverseshiftedplanar}.

\begin{lemma}
\label{le:perpiece}
	Let $G$ be a plane graph with $n$ vertices and let $P$ be a piece of $G$ 
	with $r$ vertices, $b$ boundary vertices, and $O(1)$ holes.
	In $\tilde O(bn+ b^2 r^2)$ time we can compute
	\[
		\IGL\bigl( a,V(P) \bigr) ~=~ \sum_{v\in V(P)} \, \frac{1}{d_G(a,v)} ~~~~~~
		\text{for all $a\in V(G)\setminus V(P)$}.
	\]
\end{lemma}
\begin{proof}
	First we introduce some notation, express some properties,
	and then look into the algorithmic part.
	Let $A=V(G)\setminus V(P)$, $S=\partial P$ and  
	$\Omega = A\times S$. Note that $\Omega$ has $O(bn)$ pairs.
	The set $S$ is a separator between $A$ and $V(P)\setminus S$.
	We will systematically use $a$ to index vertices of $A$, $s$ for vertices of $S$,
	and $v$ for vertices of $P$.

	Consider one fixed $a\in A$ and 
	assign weight $\delta(s)=d_G(s,a)$ to each $s\in S$.
	Consider the corresponding Voronoi diagram $\{ \cell_P(s,S)\mid s\in S\}$,
	and recall that each Voronoi cell has a dual description: 
	a family of cycles $H_P(s,S)$ in $P^*$ such that 
	\begin{align*}
		\cell_P(s,S)= V(P)\setminus \left( \bigsqcup_{\gamma\in H_P(s,S)} \exterior_P(\gamma,s) \right).
	\end{align*}	
	This Voronoi diagram depends on $a$.
	To make this dependency on $a$ explicit, because we will be considering
	several different vertices $a\in A$ together, 
	we define $X(a,s)=\cell_P(s,S)$	and $H(a,s)=H_P(s,S)$. This notation drops the
	dependency on $P$ and $S$, which we may assume constant throughout this proof.
	On the other hand, the notation makes it clear that we consider a different Voronoi diagram
	for each $a\in A$, and each such Voronoi diagram has a cell for each $s\in S$
	whose additive weight $\delta(s)$ was set to $d_G(a,s)$.	
	With the new notation we have
	\begin{equation}
		\forall (a,s)\in \Omega:~~~ X(a,s)= V(P)\setminus \left( \bigsqcup_{\gamma\in H(a,s)} \exterior_P(\gamma,s) \right).
	\label{eq:1}
	\end{equation}
	
	As we explained before, we may and we will assume that the Voronoi cells form a partition of $V(P)$.
	Therefore, for each $a\in A$, the sets $X(a,s)$, where $s\in S$, form a partition of $V(P)$.
	Since $S$ separates $A$ from $V(P)\setminus S$, it follows from the definition
	of additively-weighted Voronoi diagram that
	\begin{equation}
		\forall (a,s)\in \Omega,~ \forall v\in X(a,s): ~~~ d_G(a,v)=d_G(a,s)+d_P(s,v).
	\label{eq:2}
	\end{equation}
	Note that in the second part we are using distances in $P$, which is consistent with the distance used
	for defining the Voronoi diagram in $P$.
	The values we want to compute can be rewritten as
	\begin{align}
		\IGL\bigl(a,V(P)\bigr) 
		~&=~ \sum_{v\in V(P)} \, \frac{1}{d_G(a,v)} 
		~=~
		\sum_{s\in S}\left( \sum_{v\in X(a,s)} \,\frac{1}{d_G(a,v)}\right)
		~=~
		\sum_{s\in S}\IGL\bigl(a,X(a,s)\bigr).
	\label{eq:3}
	\end{align}
	We will explain how to compute the values $\IGL\bigl(a,X(a,s)\bigr)$
	for all $(a,s)\in \Omega$ together in $\tilde O(bn+b^2r^2)$ time. From this, the result follows,
	as the desired values can be computed in $O(|\Omega|)=O(bn)$ additional time.

	For each two distinct boundary vertices $s,s'$ of $S$, let $\tilde\Gamma(s,s')$ 
	be the family of bisectors between $s$ and $s'$.
	Note that $\tilde\Gamma(s,s')\subset \Xi(s)$ because they are bisectors.
	Define $\Gamma(s)=\bigcup_{s'\in S, s'\neq s} \tilde\Gamma(s,s')$ for each $s\in S$,
	that is, all the bisectors in $P$ between $s$ and each possible other boundary vertex.
	Obviously, $\Gamma(s)\subset \Xi(s)$.
	Each cycle in each $H(a,s)$ is described by patches from $\Gamma(s)$.
	For each $(a,s)\in \Omega$, let $k(a,s)$ be the number of patches used to describe
	all the cycles of $H(a,s)$. Since the combinatorial complexity of each Voronoi
	diagram is linear in $|S|$, we have $\sum_{s\in S} k(a,s)=O(|S|)= O(b)$ for each $a\in A$.
	
	A simple counting shows that the family $\Gamma(s)$ has in total $O(br^2)$ arcs, counted with multiplicity.
	Indeed, each family $\tilde\Gamma(s,s')$ consists of $O(r)$ cycles in the dual graph $P^*$,
	and each such cycle has $O(r)$ arcs. (Since the cycles of $\tilde\Gamma(s,s')$ are nested and any 
	two consecutive ones in the nested order differ by at least one vertex, the bound follows.)
	Thus, $\tilde\Gamma(s,s')$ has a total of $O(r^2)$ arcs, counted with multiplicity,
	which implies that $\Gamma(s)$ has in total $O(br^2)$ arcs because $|S|=b$.

	We move now to the computational part. We will consider the cycles in $H(a,s)$
	for all pairs $(a,s)\in \Omega$, first from the perspective of fibers with constant $a$,
	and then from the perspective of fibers with constant $s$. 
	
	First we compute the relevant distances.
	From each $s$ in $S$ we compute two shortest path trees, one in $G$ and one in $P$. 
	This takes $\tilde O(bn)$ time because $|S|=b$ and $G$ has $O(n)$ edges.
	Note that now, we have $d_G(s,u)$ for each $(s,u)\in S\times V(G)$ and 
	$d_P(s,v)$ for each $(s,v)\in S\times V(P)$. 
	We also compute the cycles in $\tilde\Gamma(s,s')$ for all $s,s'$. This can be done 
	in $\tilde O(r^2)$ for each $s,s'\in S$ easily; 
	see Cabello~\cite{Cabello2019} for an explicit computation
	or Gawrychowski et al.~\cite{GawrychowskiKMS21} for a faster, implicit representation.
	The idea is that a vertex $v\in V(P)$ changes sides of the Voronoi cells when
	the additive weights of the sites satisfy $\delta(s)-\delta(s')=d_P(s',v)-d_P(s,v)$.
	We easily obtain also $\Gamma(s)$ for all $s\in S$. We have spent $\tilde O(bn+b^2 r^2)$ time.

	We use Theorem~\ref{thm:VDplanar} to preprocess $P$ in $\tilde O(rb^2)$ time.
	We can do this because the vertices of $S$ are incident to the holes of $P$,
	and thus $S$ is covered by $O(1)$ faces of $P$, and those faces can
	be computed from the embedding of $G$ in $O(n)$ time.
	After this preprocessing, we can handle the following queries: 
	for any given weights $\delta(s)$ assigned to each $s\in S$ at query time,
	the additively-weighted Voronoi diagram of $S$ in $P$ can be computed
	in $\tilde O(b)$ time. Note that this Voronoi diagram uses distances in $P$
	
	We use the data structure to compute $H(a,s)$ for all $(a,s)\in \Omega$.
	More precisely, for each $a\in A$,  
	we set the weights $\delta(s)=d_G(s,a)$ for all $s\in S$, and query
	the data structure once to obtain $H(a,s)$ for all $s\in S$.
	Since we are querying the data structure $|A|<n$ times,
	we obtain $H(a,s)$ for all $(a,s)\in \Omega$ in $\tilde O(bn)$ time.
	
	Now we switch to the fibers with constant $s$.
	Consider one fixed $s\in S$.
	We assign to each vertex $v$ of $P$ weight $\omega(v)=d_P(s,v)$.
	Consider the family of cycles $\Gamma(s)$, which has $\ell=O(br)$ 
	cycles of $\Xi(b)$ with a total of $O(br^2)$ arcs, counted with multiplicity.
	Using Theorem~\ref{thm:datastructure2} for all the pairs
	\[
		(\pi,\delta) \in \bigcup_{a\in A} \Bigl( H(a,s)\times \{ d_G(s,a) \} \Bigr),
	\]
	we obtain, for each cycle $\gamma\in H(a,s)$, the value
	\begin{align*}
		\ISW_{P,s}(\gamma,d_G(a,s))~&=~ \sum_{v \in \exterior_P(\gamma,s)} \frac{1}{d_G(a,s)+\omega(v)}\\
		&=~ \sum_{v \in \exterior_P(\gamma,s)} \frac{1}{d_G(a,s)+d_P(s,v)}
	\end{align*}
	in time
	\[
		\tilde O\left( r^2 + \sum_{a\in A} k(a,s) + (br)r\right)= \tilde O\left( \sum_{a\in A} k(a,s)+br^2 \right).
	\]
	Still for the same fixed $s$ and the same weights $\omega(\cdot)$, 
	we use Lemma~\ref{le:wholepiece} with the values 
	$\{ \delta_1,\dots, \delta_{|A|}\}=\{ d_G(a,s)\mid a\in A \}$ to compute for each $a\in S$
	the value
	\[
		\sum_{v\in V(P)} \,\frac{1}{d_G(a,s) + \omega(v)} ~=~
		\sum_{v\in V(P)} \,\frac{1}{d_G(a,s) + d_P(s,v)}.
	\]
	Note that $d_G(a,s) + d_P(s,v)$ is not $d_G(a,v)$, in general.
	This takes $\tilde O(n+r)=\tilde O(n)$ time, for a fixed $s\in S$.
	
	We repeat the procedure for each $s\in S$. In total, we spend
	\[
		\sum_{s\in S} \tilde O\left( \sum_{a\in A} k(a,s)+br^2 \right) ~=~
		\tilde O\left( |S| br^2 + \sum_{a\in A} \sum_{s\in S} k(a,s) \right) ~=~
		\tilde O\left( b^2r^2 + \sum_{a\in A} b \right) ~=~
		\tilde O\left( b^2r^2 + bn \right)
	\]
	time and we obtain the values
	\begin{align*}
		\sigma(\gamma,a,s) ~&:=~
		\sum_{v \in \exterior_P(\gamma,s)} \frac{1}{d_G(a,s)+d_P(s,v)}\\
		\tau(a,s) ~&:=~
		\sum_{v \in V(P)} \frac{1}{d_G(a,s)+d_P(s,v)}
	\end{align*}
	for all $\gamma\in H(a,s)$ and all $(a,s)\in \Omega$.
	
	Combining properties~\eqref{eq:1} and~\eqref{eq:2}, we have for each $(a,s)\in \Omega$
	\begin{align*}
		\IGL\bigl(a,X(a,s)\bigr) 
		~&=~
		\sum_{v\in X(a,s)} \, \frac{1}{d_G(a,v)}\\
		&=~ \sum_{v\in X(a,s)} \, \frac{1}{d_G(a,s)+d_P(s,v)}\\
		&=~ \sum_{v\in V(P)} \, \frac{1}{d_G(a,s)+d_P(s,v)} -
			\sum_{\gamma\in H(a,s)} \left( \sum_{v\in \exterior_P(\gamma,s)} \, \frac{1}{d_G(a,s)+d_P(s,v)}  \right)\\
		&=~ \tau(a,s) - \sum_{\gamma\in H(a,s)} \sigma(\gamma,a,s).
	\end{align*}
	Thus, we can obtain each single $\IGL\bigl(a,X(a,s) \bigr)$ in 
	$O(1+|H(a,s)|)$ time, using the data already computed. 
	Since $\sum_{s\in S} |H(a,s)|\le \sum_{s\in S} k(a,s)= O(b)$, this computation 
	for all $(a,s)\in \Omega$ together takes $O(bn)$ time.
	With this we have computed $\IGL\bigl(a,X(a,s)\bigr)$
	for all $(a,s)\in \Omega$, and we have finished the proof because of equation~\eqref{eq:3}.
\end{proof}


\subsection{Final algorithm}
Using Lemma~\ref{le:perpiece} for each piece in a $r$-division with a few holes,
we obtain the final result.

\begin{theorem}
	Let $G$ be a planar graph with $n$ vertices and positive, abstract edge lengths.
	In $\tilde O(n^{9/5})$ time we can compute 
	\[
		\IGL(a,V(G)\setminus\{ a \}) ~=~ \sum_{v\in V(G)\setminus\{ a\}} \frac{1}{d_G(a,v)}
		~~~~~~\text{for all $a\in V(G)$}.
	\]
\end{theorem}
\begin{proof}
	First, we compute an $r$-division of $G$ with few holes.
	Such division has $k=O(n/r)$ pieces $\PP=\{ P_1,\dots,P_k\}$, 
	where each piece $P_i\in \PP$ has $O(r)$ vertices,
	$O(\sqrt{r})$ boundary vertices, and $O(1)$ holes. 
	Moreover, different pieces only intersect at boundary vertices.
	Such an $r$-division can be computed in $\tilde O(n)$ time~\cite{FakcharoenpholR06,KleinMS13,KleinS98}. 
	
	Let $B$ be the set of vertices that are boundary vertices in some piece.
	Note that $B$ contains $O(n/r)\cdot O(\sqrt{r})= O(n/\sqrt{r})$ vertices.
	The vertex set $V(G)$ is the disjoint union of $B, V(P_1)\setminus \partial P_1, 
	\dots, V(P_k)\setminus \partial P_k$.
	
	We compute a shortest path tree from each vertex of $B$. Each shortest
	path tree takes linear time~\cite{hkrs-fspap-97}, 
	and thus we spend $O(n^2/\sqrt{r})$ time in this computation.
	After this, for each vertex $u$ of $G$ and each vertex $b$ of $B$, we
	can retrieve $d_G(u,b)$ in constant time.
	In particular, any sum of inverse of distances between pair of vertices, 
	where are at least one vertex is in $B$, can be computed in time proportional to the number
	of terms in the sum.
	It follows that the desired values $\IGL(a,V(G)\setminus\{ a \})$ can be
	computed now for all $a\in B$ with additional time $O(|B|\cdot |V(G)|)=O(n^2/\sqrt{r})$.
	
	Consider one fixed piece $P_i\in \PP$.
	We use Lemma~\ref{le:perpiece}
	to compute in $\tilde O(r^{1/2} n+ r^3)$ time the values
	$\IGL\bigl(a,V(P_i) \bigr)$ for all $a\in V(G)\setminus V(P_i)$.
	We also compute $\IGL(a, \partial P_i)$ by adding the 
	distances from $\partial P_i\subset B$; this takes $O(r^{1/2} n)$ time.
	With this we obtain the values
	\[
		\IGL(a,V(P_i)\setminus\partial P_i)~=~ \IGL(a,V(P_i)) - \IGL(a,\partial P_i)
		~~~~~~\text{for all $a\in V(G)\setminus V(P_i)$}.
	\]
	Finally, we compute the distance (in $G$) between all pairs of vertices in $P_i$
	in $\tilde O(r^2)$ using usual machinery. For example, we may add edges
	between the $O((\sqrt{r})^2)=O(r)$ pairs of boundary vertices $\partial P_i$
	and compute distances in the resulting graph, which has $O(r)$ vertices and edges.
	Repeating this for the $O(n/r)$ pieces of $\PP$ together we have spent $\tilde O(n^2/\sqrt{r} + nr^2)$ time.

	Recall that $V(G)$ is the disjoint union of $B, V(P_1)\setminus \partial P_1, 
	\dots, V(P_k)\setminus \partial P_k$. Thus, for 
	each $a\in V(G)\setminus B$, there is a unique index $j(a)$ such that $a\in P_{j(a)}\in \PP$,
	and therefore 
	\[
		\IGL(a,V(G)\setminus\{ a\}) ~=~ 
			\sum_{i\in[k]\setminus\{ j(a) \}} \IGL(a,V(P_i)\setminus\partial P_i)
			~+~ \sum_{b\in B} \frac{1}{d_G(a,b)}
			~+~ \sum_{v\in V(P_{j(a)})\setminus \partial P_{j(a)}} \frac{1}{d_G(a,v)} \, .
	\]
	The $k-1=O(n/r)$ terms in the first sum have been computed, 
	the $|B|=O(n/\sqrt{r})$ terms in the second sum have also been computed,
	and the $|V(P_{j(a)})\setminus \partial P_{j(a)}|=O(r)$ terms in the
	last sum have also been computed.
	We conclude that, for a single $a\in V(G)\setminus B$, we can recover
	$\IGL(a,V(G)\setminus\{ a\})$ in $O((n/\sqrt{r}) + r)$ time.
	Doing this for each vertex $a\in V(G)\setminus B$ takes $O(n^2/\sqrt{r} + nr)$.
	
	The whole algorithm, as explained, takes $\tilde O(n^2/\sqrt{r} + nr^2)$.
	Setting $r=n^{2/5}$ we obtain a running time of $\tilde O(n^{9/5})$.
\end{proof}

\begin{corollary}
	Let $G$ be a planar graph with $n$ vertices and positive, abstract edge lengths.
	In $\tilde O(n^{9/5})$ time we can compute $\IGL(G)$.
\end{corollary}
\begin{proof}
	We use the Theorem to compute $\IGL(a,V(G)\setminus \{a \})$ for
	all $a\in V(G)$. Then we note that
	\[
		\IGL(G) ~=~ \frac 12 \sum_{a\in V(G)} \IGL(a,V(G)\setminus \{a \}).
	\]
\end{proof}

\section{Conclusions}

Let us remark a few extensions of the results we presented. The extensions
are applicable for graphs of bounded treewidth and for planar graphs without
affecting the asymptotic running times of each case.

Firstly, the algorithm can be adapted to work for directed planar graphs.
In such case it is meaningful to compute
\[
	\sum_{uv\in \binom{V(G)}{2}} \left( \frac{1}{d_G(u,v)}+\frac{1}{d_G(v,u)}\right).
\]
We provide two clarifying remarks on how to achieve this. 
The first observation is that distances \emph{to} a separator vertex can be 
computed by reversing the arcs and computing the distances \emph{from} the separator vertex.
The second observation is that Voronoi diagrams can also be considered for such "distances"
defined in directed graphs; this has been done also in previous works~\cite{Cabello2019,GawrychowskiKMS21}.

As a second extension, we note that we can also compute the sum of the distances between some
\emph{marked} vertices. This means that, for any given $U\subset V(G)$, 
we can compute in the same asymptotic running time the sum
\[
	\sum_{uv\in \binom{U}{2}} \frac{1}{d_G(u,v)}
	~~~~~~\text{ or }~~~~~~
	\sum_{u\in U}\sum_{v\in V(G)\setminus\{u\}} \frac{1}{d_G(u,v)}.
\]
For the case of planar graphs, where this may be less obvious, 
see the remark after Theorem~\ref{thm:datastructure2}.

As a final extension, note that we can handle any rational function of constant degree
that depends on the distances.
For example, the same approach can be used for computing
\[
	\sum_{uv\in \binom{V(G)}{2}} \frac{d_G(u,v)}{(1+ d_G(u,v))^2}
	~~~~~~\text{ or }~~~~~~
	\sum_{uv\in \binom{V(G)}{2}} (d_G(u,v))^2.  
\]
Indeed, this only affects the rational (or polynomial) functions that we have to consider 
and evaluate (at the \emph{shifts}),
but Lemma~\ref{lem:multipoint_evaluation} keeps being applicable.
On the other hand, current methods developed here or in previous works
do not seem applicable to compute in subquadratic time values like
\[
	\sum_{uv\in \binom{V(G)}{2}} \sqrt{d_G(u,v)}  
	~~~~~~\text{ or }~~~~~~
	\sum_{uv\in \binom{V(G)}{2}} \frac{1}{\sqrt{d_G(u,v)}} , 
\]
even assuming a strong model of computation, like Real RAM, where sums
of square roots can be manipulated in constant time.

\section*{Acknowledgments}
I am grateful to the reviewers for pointing out that the model of computation should be clarified
and several additional corrections.

\bibliographystyle{abuser}
\bibliography{bibliography-IGL}

\end{document}